\documentclass{amsart}

\usepackage[T1]{fontenc}
\usepackage{latexsym,amsfonts}
\usepackage{tikz}
\usetikzlibrary{calc,shapes,arrows,automata}
\usepackage{amsmath, amssymb}
\usepackage{algorithm}
\usepackage{algorithmic}
\usepackage{url}
\usepackage{subcaption} 
\usepackage{caption}
\usepackage{centernot}
\usepackage{epsfig}
\usepackage{thmtools}
\usepackage{hyperref}
\usepackage{multirow}
\usepackage{textcomp}
\usepackage{supertabular}
\usepackage{preamble}

\begin{document}

\begin{abstract}
\label{sec:abstract}
The recently introduced notions of ranking functions and closure certificates utilize well-foundedness arguments to facilitate the verification of dynamical systems against $\omega$-regular properties.
A \emph{ranking function} and a \emph{closure certificate} are real-valued functions defined over states and state pairs of a dynamical system whose zero superlevel sets are inductive state invariant and inductive transition invariant, respectively.  
The search for such certificates can be automated by fixing a specific template class, such as a polynomial of a fixed degree, and then using optimization techniques such as sum-of-squares (SOS) programming  to find it. 
Unfortunately, such certificates may not be found for a fixed template.
In such a case, one must change the template; for example, increase the degree of the polynomial.
In this paper, we consider a notion of multiple functions in the form of vector certificates.
Taking inspiration from the literature on vector barrier certificates as generalizations of standard barrier certificates for safety verification, we propose vector co-B\"uchi ranking functions and vector closure certificates as nontrivial generalizations of ranking functions and closure certificates, respectively.
Both notions consist of a set of functions that jointly overapproximate an inductive invariant by considering each function to be a linear combination of the others.
The advantage of such certificates is that they allow us to prove properties even when a single function for a fixed template cannot be found using standard approaches. 
We present an SOS programming approach to search for these functions and demonstrate the effectiveness of our proposed method in verifying $\omega$-regular specifications in several case studies.
\end{abstract}

\title[Vector Certificates for $\omega$-regular Specifications]{Vector Certificates for $\omega$-regular Specifications}

\thanks{
This work was supported by NSF grants CNS-2145184 and CNS-2111688.}

\author[mohammed adib oumer]{Mohammed Adib Oumer} 
\author[vishnu murali]{Vishnu Murali} 
\author[majid zamani]{Majid Zamani}

\address{Department of Computer Science at the University of Colorado, Boulder, CO, USA.}
\email{\{mohammed.oumer,~vishnu.murali,~majid.zamani\}@colorado.edu}
\urladdr{https://www.hyconsys.com/members/moumer/}
\urladdr{https://www.hyconsys.com/members/vmurali/}
\urladdr{https://www.hyconsys.com/members/mzamani/}

\maketitle

\section{Introduction} 
\label{sec:intro}

An automated abstraction-free approach to verify systems against specifications characterized in linear temporal logic (LTL)~\cite{pnueli1977temporal} or via $\omega$-regular automata~\cite{vardi1994reasoning} is through the use of ranking functions~\cite{chatterjee2024sound} or closure certificates \cite{murali2024closure}, which act as inductive transition invariants.
Ranking functions overestimate reachable states and use well-foundedness arguments to verify $\omega$-regular specifications.
On the other hand, the closure certificates introduced in \cite{murali2024closure} aim to capture the reachable transitions of a system with the goal of automating the verification of LTL and $\omega$-regular specifications.
 The search for such certificates can be automated using optimization \cite{parrilo2003semidefinite,prajna2002sostools}, neural networks \cite{nadali2024closure}, or SMT-based approaches \cite{moura2011smt}. 
These approaches typically rely on first fixing a template (e.g., a polynomial of a fixed degree) and then making use of the above approaches to search for an appropriate function within this template. 
When one is unable to find such a certificate for a given template, one is forced to consider a different template and try again. 
Inspired by the use of a vector of multiple functions in vector barrier certificates \cite{sogokon2018vector} and its powerful potential for generalization, we present notions of \emph{vector co-B\"uchi ranking functions} and \emph{vector closure certificates}.
In this paper, we show that one may find such certificates to verify LTL properties, such as safety and persistence, even when we fail to find a single function of the same template to do so. 
Similarly to vector barrier certificates, the search for vector co-B\"uchi ranking functions and vector closure certificates can be carried out effectively using existing approaches.

Verification of systems against safety, a prominent LTL property, can be automated by using barrier certificates~\cite{prajna2004safety}.
A barrier certificate is a real valued function that is nonpositive over the initial states, positive over the unsafe states, and nonincreasing with transitions. 
Thus, a barrier certificate is an inductive state invariant that guarantees safety, as its zero level set separates the reachable and unsafe states, with the zero-sublevel set overapproximating the set of reachable states.
In~\cite{chatterjee2024sound}, the authors extended this safety verification approach via inductive invariants and unified it with well-foundedness arguments \cite{alpern1987recognizing} to certify general LTL properties.
In this method, they argue for well-foundedness by considering the existence of a single ranking function.
However, if one wants to consider partitions over relevant sets instead and try to find independent ranking function arguments over each partition, then one should consider a well-foundedness argument over the reachable transitions of the system.

Transition invariants~\cite{podelski2004transition} were introduced as an overapproximation of reachable transitions and are used to verify programs against $\omega$-regular properties. 
Closure certificates \cite{murali2024closure} act as functional analogs of inductive transition invariants. 
Such certificates can be used to verify safety, as well as larger classes of $\omega$-regular properties.
Intuitively, a closure certificate is a real valued function that is nonnegative for a pair of states $(x,y)$ if $y$ may be reachable from $x$ (conversely, if the function is negative, then $y$ is not reachable from $x$).
Ensuring that this certificate is negative for all pairs $(x_0, x_u)$, where $x_0$ is an initial state and $x_u$ is an unsafe state, acts as a proof of safety.
To verify general LTL specifications, one may extend the earlier findings with a ranking function argument to prove the well-foundedness.

In the context of dynamical systems, the typical formulation for ranking functions and closure certificates assumes the presence of a single function.
This formulation is sound, but if one fails to find a function that meets a given set of conditions, it does not imply that the system refutes the property of interest for the system.
An approach to addressing this challenge is to formulate the conditions using a systematic combination of multiple functions.
The authors of~\cite{sogokon2018vector} introduced the use of a vector of functions for barrier certificates in what they termed vector barrier certificates.
They show that such a formulation generalizes the standard formulation of barrier certificates, with guarantees that sometimes a scalar barrier certificate of a given template does not exist, while a vector barrier certificate does for a given system.
In addition, the class of inductive state invariants that one can extract from vector barrier certificates is richer, as its individual elements do not need to correspond to an inductive invariant. 
In this paper, we adopt the vector-based formulation for ranking functions and closure certificates and show its benefits in verifying $\omega$-regular properties, both in terms of the certificate template as well as computation time. 
The benefit of vector closure certificates over vector co-B\"uchi ranking functions can be summarized as follows.
As noted in~\cite{podelski2004transition}, a relation is well-founded if its transitive closure is disjunctively well-founded (not if the relation itself is disjunctively well-founded). 
Since we are using multiple functions to overapproximate the transitive closure of the transition relation via vector closure certificates, we can use the functions in the vector such that each component decreases independently of the others over different partitions of relevant sets to prove well-foundedness.
 This approach is not possible for vector co-B\"uchi ranking functions. 
In contrast, vector co-B\"uchi ranking functions are defined over states of a system, making their formulation simpler than vector closure certificates that are defined over pairs of states.

\subsection{Contributions}
In this paper, our main contributions are as follows:
i) We propose novel notions of vector co-B\"uchi ranking functions and vector closure certificates that have more permissive conditions compared to the standard formulations;
ii) We present SOS programming to search for the proposed vector certificates using a fixed template;
iii) We show that when a standard closure certificate of a given template fails to ensure a certain specification, one can find functions of the same template using the proposed vector closure certificates;
iv) We show that vector co-B\"uchi ranking functions are ideal for minimal computation times compared to other certificates;
v) We show that even when a standard closure certificate of a given template is successfully computed to ensure a given specification, one can find functions of a lower degree template using the proposed vector closure certificates while drastically improving computation time;
vi) We show that, for suitable choices of certain hyperparameters, our proposed formulations recover the usual conditions for ranking functions and closure certificates.

\subsection{Related Works}
The authors of~\cite{prajna2004safety} proposed the notion of barrier certificates as a discretization-free approach to provide guarantees of safety~\cite{prajna2007convex} for dynamical and hybrid systems. 
The authors of~\cite{kong2013exponential} introduced exponential-type barrier certificates that generalize the conditions of~\cite{prajna2004safety} while maintaining the convexity of the conditions simply by fixing a hyperparameter.
In~\cite{dai2017barrier}, general barrier certificates are reported, which further generalize the conditions for exponential-type barrier certificates. 
The authors of~\cite{sogokon2018vector} used a vector of functions to 
generalize the conditions of~\cite{kong2013exponential} and~\cite{dai2017barrier} for the verification of safety. 
Another generalization of standard barrier certificates using multiple functions was introduced in~\cite{oumer2024ibc} using the notion of logical interpolation.

The authors of \cite{chatterjee2024sound,dimitrova2014deductive} consider the use of barrier certificates and ranking functions to verify LTL properties.
The authors of~\cite{podelski2004transition} proposed a notion of transition invariants and demonstrated their use in verifying programs against $\omega$-regular properties. 
This concept was adopted in~\cite{murali2024closure} to introduce the notion of closure certificates for an automated approach to verify dynamical systems against $\omega$-regular properties. 

In terms of searching for vector certificates, it is not generally as straightforward as the standard certificates due to relatively more complex bilinear conditions (cf. Section~\ref{subsec:specifications} and~\ref{sec:vcs}).
One approach to addressing this bilinearity, albeit conservative, is to fix one of the terms contributing to the bilinear terms to reduce the problem to a convex one and use existing computational tools such as SOS programming or SMT solvers to search for the certificates. 
For the standard certificates, reducing the bilinear conditions to linear ones requires fixing only one constant.
In contrast, for vector certificates, we want to consider combinations of the functions; thus, reducing the bilinear conditions to linear ones requires fixing a matrix that is not trivially the zero matrix or the identity.
Another approach that directly tackles this challenge uses an iterative scheme of fixing one of the bilinear terms and searching for the other terms until an optimal solution is obtained or the process times out, as shown in~\cite{berger2024cone,wang2021synthesizing}. However, this approach generally requires a feasible solution a priori to start the iterative scheme. In our approach, our aim is to find a feasible solution to begin with.

\section{Preliminaries}
\label{sec:prelims}

\subsection{Notation}
\label{subsec:notation}
We use $\N$ and $\R$ to denote the sets of natural numbers and reals, respectively. 
For $m \in \R$, we use $\R_{\geq m}$ and $\R_{> m}$ to denote the intervals $[m, \infty)$ and $(m,\infty)$, respectively. 
Similarly, for any natural number $n \in \N$, we use $\N_{\geq n}$ to denote the set of natural numbers greater than or equal to $n$.
The $n$-dimensional Euclidean space is denoted by $\mathbb{R}^n$. 
$\1_n$ denotes a vector of $n$ ones.
We say that a matrix $A \in \R^{m\times n}$ is nonnegative if all its elements are nonnegative. 
We denote the set of nonnegative matrices using $\R_{\geq 0}^{m\times n}$.
Symbols $\forall$ and $\exists$ denote the universal and existential quantifiers, respectively.

Given a set $X$, $|X|$ denotes the cardinality of the set. 
Given a relation $R \subseteq X \times Y$ and an element $x \in X$, we use $R(x)$ to denote the set $\{ y \mid (x,y) \in R \}$. 
The set $X^{\omega}$ denotes the set of countably infinite sequences of elements in $X$.
We use the notation $(x_1, x_2, \ldots, x_n) \in X^{*}$ for finite length sequences and $(x_0, x_1, \ldots )\in X^{\omega}$ for $\omega$-sequences.
Let $\Inf(s)$ be the set of elements that occur infinitely often in the sequence $s  =(x_0, x_1, \ldots )$.
Given an infinite sequence $s = (x_0, x_1, \ldots )$ and two natural numbers $i,j \in \N$ where $i \leq j$, we use $s[i,j]$ to indicate the finite sequence $(x_i, x_{i+1}, \ldots, x_j)$, and $s[i, \infty)$ to indicate the infinite sequence $( x_i, x_{i+1}, \ldots )$. 
Finally, we use $s[i]$ to denote the $i^{th}$ element in the sequence $s$, \textit{i.e.}, we have $s[i] = x_i$ for any $i \in \N$.

\subsection{Discrete-time Dynamical System}
\label{subsec:system}
In this paper, we model systems as discrete-time dynamical systems.
\begin{definition}
\label{def:system}
    A discrete-time dynamical system is given by the tuple $\Sys = (\Xx, \Xx_0, f)$,
    where the state set is denoted by $\Xx$, the set of initial states is $\Xx_0 \subseteq \Xx$, and $f: \Xx \times \Xx$ is the state transition relation that describes the evolution of the states of the system. 
    That is, for $x_t$, the state of the system at time step $t\in \N$, the state of the system in the next time step is given by $x_{t+1} \in f(x_t)$, $\forall x_t \in \Xx.$
\end{definition}

If $\forall x \in \Xx$, we have $|f(x)| = 1$, then we consider the transition relation $f$ to be a deterministic state transition function.
We use $f$ for both a set-valued map when it is a relation and a transition function when it is a function.
Throughout the paper, we assume that the state sets of the systems under consideration are compact.

For notational convenience, given a state $x \in \Xx$, we use $x'$ to indicate a state in $f(x)$ (\textit{i.e.}, $x' \in f(x)$), unless otherwise specified. 
Furthermore, we use $x_0 \in \Xx_0$ to denote a state in the initial set of states $\Xx_0$.
A state sequence is an infinite sequence  $(x_0, x_1, \ldots) \in \Xx^{\omega}$ where $x_0 \in \Xx_0$, and $x_{i+1} \in f(x_i)\ \forall i \in \N$.
We say $x_j$ is reachable from $x_i$ if $i,j \in \N, j > i$ and $x_i, x_j \in (x_0, x_1, \ldots)$ for some state sequence $(x_0, x_1, \ldots)$.

\subsection{Specifications}
\label{subsec:specifications}
In this paper, we study increasingly complex specifications, from safety and persistence to LTL and $\omega$-regular specifications over discrete-time dynamical systems.

\subsubsection{Safety}
We say that a system $\Sys$ starting from some set of initial states $\Xx_0$ is safe with respect to a set of unsafe states $\Xx_u \subseteq \Xx$ if, for any state sequence $(x_0, x_1, \ldots)$, $x_i \notin \Xx_u\ \forall i \in \N_{\geq 1}$.
Remark that we assume that the initial and unsafe sets are disjoint, i.e. $\Xx_0 \cap \Xx_u = \emptyset$.
For notational simplicity, we use $x_u \in \Xx_u$ to denote a state in the unsafe set of states $\Xx_u$. 
\begin{definition}[BC]
\label{def:bc}
   A function $\Bb: \Xx \to \R$ is a BC for a system $\Sys$ with respect to a set of unsafe states $\Xx_u$ if $\exists \lambda \in \R_{\geq 0}$ such that $\forall x \in \Xx$, $\forall x_0 \in \Xx_0$, $\forall x_u \in \Xx_u$, $\forall x' \in f(x)$:
    \begin{align}
        \label{eq:bc_1}
        & \Bb(x_0) \leq 0, \\
        \label{eq:bc_2}
        & \Bb(x_u) > 0, \text{ and } \\
        \label{eq:bc_3}
        & \Bb(x') \leq \lambda \Bb(x).
    \end{align}
\end{definition}
\begin{theorem}[BCs imply safety~\cite{prajna2004safety}]
    For a system $\Sys = (\Xx, \Xx_0, f)$ with unsafe states $\Xx_u$, the existence of a barrier certificate $\Bb$ satisfying conditions~\eqref{eq:bc_1}-\eqref{eq:bc_3} implies its safety.
\end{theorem}

The definition of \emph{vector barrier certificates (VBCs)}, adopted from \cite{sogokon2018vector}, for discrete-time systems is given as follows.
\begin{definition}[VBC]
\label{def:vbc}
A vector of functions $\B$ consisting of $\Bb_i: \Xx \rightarrow \R$,  $\forall  i \in \{1,\ldots, k \}$, is a VBC for a system $\Sys$ with respect to a set of unsafe states $\Xx_u$ if $\exists A \in \R_{\geq 0}^{k\times k}$ such that $\forall x \in \Xx$, $\forall x_0 \in \Xx_0$, $\forall x_u \in \Xx_u$, $\forall x' \in f(x)$:
\begin{align}
    \label{eq:vbc_1}
    &\B (x_0) \leq 0, && 
    \\
    \label{eq:vbc_2}
    &\bigvee_{i = 1}^{k} \big(\Bb_i(x_u) > 0 \big),  
    \\
    \label{eq:vbc_3}
    &\B(x') \leq A\B(x),
\end{align}
where the inequalities are element-wise.
\end{definition}
\begin{theorem}[VBCs imply safety~\cite{sogokon2018vector}]
The existence of a VBC $\Bb_i: \Xx \rightarrow \R$, $\forall i \in \{1,\ldots,k \}$, satisfying conditions~\eqref{eq:vbc_1}-\eqref{eq:vbc_3} for a system $\Sys = (\Xx, \Xx_0, f)$ with unsafe states $\Xx_u$ implies its safety.
\label{thm:vbc}
\end{theorem}

Note that conditions~\eqref{eq:bc_1} and~\eqref{eq:bc_3} (similarly~\eqref{eq:vbc_1} and~\eqref{eq:vbc_3}) provide the invariant set that overapproximates the reachable states, while condition~\eqref{eq:bc_2} (similarly~\eqref{eq:vbc_2}) separates the unsafe states from the reachable states.

Now, we discuss the existing formulations of the closure certificates.
The formulation of closure certificates (CCs)~\cite{murali2024closure} for safety is given below.
\begin{definition}[CC for Safety]
\label{def:cc_safe}
A  function $\Tt: \Xx \times \Xx \to \R$ is a CC for a system  $\Sys = (\Xx, \Xx_0, f)$ with respect to a set of unsafe states $\Xx_{u}$ if $\exists \eta \in \R_{ > 0}, \lambda \in \R_{ \geq 0}$ such that $\forall x, y \in \Xx$, $\forall x_0 \in \Xx_0$, $\forall x_u \in \Xx_u$, $\forall x' \in f(x)$:
\begin{align}
    \label{eq:cc_safe1} 
    & \Tt(x, x') \geq 0, \\
    \label{eq:cc_safe2}
    & \Tt(x, y) \geq \lambda \Tt(x', y), \text{ and } \\
    \label{eq:cc_safe3}
    & \Tt(x_0, x_u) \leq - \eta. 
\end{align}
\end{definition}
\begin{theorem}[CCs imply safety \cite{murali2024closure}]
    \label{thm:cc_safe}
   The existence of a function $\Tt: \Xx \times \Xx \to \R$ satisfying conditions~\eqref{eq:cc_safe1}-\eqref{eq:cc_safe3} for a system $\Sys = (\Xx, \Xx_0, f)$ with unsafe states $\Xx_{u}$ implies its safety.
\end{theorem}

\subsubsection{Persistence}
We say that a system visits a region $\Xx_{VF} \subseteq \Xx$ only finitely often if, for any state sequence $( x_0, x_1, \ldots)$, $\exists i \in \N$ such that $\forall j \geq i$, $j \in \N$, we have $x_j \notin \Xx_{VF}$.
An approach to ensure persistence is the use of a Lyapunov-like argument \cite{podelski2006model} combined with an invariant argument via CCs.
\begin{definition}[CC for Persistence]
\label{def:cc_persistence}
    A function $\Tt: \Xx \times \Xx \to \R$ is a CC for $\Sys = (\Xx, \Xx_0, f)$ with a set of states $\Xx_{VF} \subseteq \Xx$ that must be visited only finitely often if $\exists \eta \in \R_{>0}, \lambda \in \R_{\geq 0}$ such that $\forall x,y \in \Xx$, $\forall x_0 \in \Xx_0$, $\forall x' \in f(x)$ and $\forall y', y'' \in \Xx_{VF}$: 
    \begin{align}
        \label{eq:cc_pers1}
        &\Tt(x, x') \geq 0, \\
        \label{eq:cc_pers2}
        &\Tt(x, y) \geq \lambda \Tt(x', y), \text{ and }\\
        \label{eq:cc_pers3}
        &\big( \Tt(x_0, y') \geq 0 \big) \wedge \big( \Tt(y',y'') \geq 0 \big) \implies \nonumber \\
        &\qquad \big( \Tt(x_0,y'') \leq \Tt(x_0,y') - \eta \big). 
    \end{align}
\end{definition}
\begin{theorem}[CCs imply Persistence \cite{murali2024closure}]
    \label{thm:cc_persistence}
    The existence of a function $\Tt: \Xx \times \Xx \to \R$ satisfying conditions~\eqref{eq:cc_pers1}-\eqref{eq:cc_pers3} for a system $\Sys = (\Xx, \Xx_0, f)$ implies that the state sequences of the system visit the set $\Xx_{VF}$ finitely often.
\end{theorem}

Observe that conditions~\eqref{eq:cc_pers1} and~\eqref{eq:cc_pers2} provide an invariant set that overapproximates the reachable transitions while condition~\eqref{eq:cc_pers3} argues for finite visits using possible transitions.
One can employ the S-procedure~\cite{yakubovich1971s} to convert the implication in~\eqref{eq:cc_pers3} into a single inequality condition.

We now mention the formulation of B\"uchi ranking functions (BRF) to argue for infinite visits (recurrence). 
\begin{definition}[BRF for Recurrence]
\label{def:brf} 
A function $\Bb$ is a BRF for a system $\Sys$ with respect to a set of states $\Xx_{INF} \subseteq \Xx$ that must be visited infinitely often if $\exists \eta \in \R_{> 0}$ and $\exists \lambda_1 \in \R_{\geq 0}$ such that $\forall x \in \Xx$, $\forall x_0 \in \Xx_0$, $\forall x' \in f(x)$, $\forall y \in \Xx \backslash \Xx_{INF}$, $\forall y' \in f(y)$:
\begin{align}
    \label{eq:brf_1}
    &\Bb(x_0) \geq 0, && \\
    \label{eq:brf_2}
    &\Bb(x') \geq \lambda_1 \Bb(x), \text{ and} \\
    \label{eq:brf_3}
    &(\Bb(y) \geq 0) \implies (\Bb(y) - \Bb(y') - \eta \geq 0),
\end{align}
where the inequalities are element-wise.
\end{definition}
\begin{theorem}[BRFs imply recurrence~\cite{chatterjee2024sound}]
\label{thm:brf}
The existence of a BRF $\Bb$, satisfying conditions~\eqref{eq:brf_1}-\eqref{eq:brf_3} for a system $\Sys = (\Xx, \Xx_0, f)$ implies that the state sequences of the system visit the set $\Xx_{INF}$ infinitely often.
\end{theorem}

In this paper, we modify this result to argue for a finite visit to the set $\Xx_{VF}$. We replace condition~\eqref{eq:brf_3} with conditions $(\Bb(y) \geq 0) \implies (\Bb(y) - \Bb(y') \geq 0)$ and $(\Bb(z) \geq 0) \implies (\Bb(z) - \Bb(z') - \eta \geq 0)$ $\forall y \in \Xx \backslash \Xx_{VF}$, $\forall y' \in f(y)$, $\forall z \in \Xx_{VF}$, $\forall z' \in f(z)$. 
The key idea behind these conditions is to set the function $\Bb$ to be nonincreasing when outside $\Xx_{VF}$, but to make it strictly decreasing when outside $\Xx_{VF}$. 
Thus, if such $\Bb$ exists, $\Xx_{VF}$ must be visited only finitely often.
We explore a vector certificate version of this formulation in Section~\ref{sec:vcs}.

\subsubsection{Linear Temporal Logic (LTL)}
The formulae in LTL~\cite{pnueli1977temporal}  are defined with respect to a set of finite atomic propositions $AP$ that are relevant to our system. 
Let $\Sigma = 2^{AP}$ denote the powerset of atomic propositions.
A trace $w = ( w_0, w_1, \ldots, ) \in \Sigma^{\omega}$ is an infinite sequence of sets of atomic propositions.
The syntax of LTL can be given via the following grammar: 
\[
\phi := \top \;|\; a \;|\; \neg \phi \;|\; \mathsf{X} \phi  \;|\; \phi \mathsf{U} \phi ,
\]
where $\top$ indicates $\mathsf{true}$, $a \in AP$ denotes an atomic proposition, the symbols $\wedge$ and $\neg$ denote the logical AND and NOT operators, respectively.
The temporal operators next and until are denoted by $\mathsf{X}$ and $\mathsf{U}$, respectively.
The above operators are sufficient to derive the logical OR ($\vee$) and implication ($\implies$) operators, as well as the temporal operators release ($\mathsf{R}$), eventually ($\mathsf{F}$), and always ($\mathsf{G}$).

We inductively define the semantics of an LTL formula with respect to the trace $w$ as follows:
\begin{align}
&w \models a && \text{ if } a \in w[0] \\
& w \models \phi_1 \wedge \phi_2 && \text{ if } w \models \phi_1 \text{ and } w \models \phi_2 \\
& w \models \neg \phi && \text{ if } w \not\models \phi \\
& w \models \mathsf{X} \phi & & \text{ if } w[1, \infty) \models \phi\\
& w \models \phi_1 \mathsf{U} \phi_2 && \text{ if }\exists i \in \N \text{ such that } w[0,i] \models \phi_1 \nonumber \\ & && \text{ and } w[i+1, \infty) \models \phi_2
\end{align}
To reason about whether a system satisfies a property specified in LTL, we associate a labeling function $\Ll: \Xx \to \Sigma$ that maps each state of the system to a letter in the finite alphabet $\Sigma$. 
This naturally generalizes to mapping a state sequence of the system $( x_0, x_1, \ldots, ) \in \Xx^{\omega}$ to a trace $w = ( \Ll(x_0), \Ll(x_1), \ldots, ) \in \Sigma^{\omega}$. 
Let $TR(\Sys, \Ll)$ denote the set of all traces of $\Sys$ under the labeling map $\Ll$. 
Then the system $\Sys$ satisfies an LTL property $\phi$ under the labeling map $\Ll$ if, for all $w \in TR(\Sys, \Ll)$, we have $w \models \phi$.
We denote this as $\Sys \models_{\Ll} \phi$ and infer the labeling map from context.
Safety and persistence can be formulated as LTL formulae.

\vspace{0.3em}\noindent \textbf{Nondeterminstic B\"uchi Automata.}
A nondeterminstic B\"uchi automaton (NBA) $\Aa$ is a tuple $(\Sigma,\Qq, \Qq_0, \delta, \Qq_F)$, where:
$\Sigma$ is the alphabet, 
$\Qq$ is a finite set of states,
$\Qq_0 \subseteq \Qq$ is an initial set of states, 
$\delta : \Qq \times \Sigma \times \Qq$ is the transition relation and 
${\Qq_F} \subseteq \Qq$ is the set of accepting states.
A run of the automaton $\Aa = (\Sigma,\Qq, \Qq_0, \delta, \Qq_F)$ over a trace $w = ( \sigma_0, \sigma_1, \sigma_2 \ldots, ) \in \Sigma^{\omega}$ is an infinite sequence of states  characterized as $\rho = ( q_0,q_1, q_2, \ldots, ) \in \Qq^{\omega}$ with $q_0 \in \Qq_0$ and $q_{i+1} \in \delta(q_i, \sigma_i)$.
An NBA $\Aa$ is said to accept a trace $w$ if there exists a run $\rho$ over $w$ where $\Inf(\rho) \cap {\Qq_F} \neq \emptyset$.

It is well known that given an LTL formula $\phi$ on a set of atomic propositions $AP$, one can construct an NBA $\Aa$ such that $w \in \Ll(\Aa)$ if and only if $w \models \phi$~\cite{vardi2005automata}.
An automata-theoretic approach for checking whether $\Sys \models_{\Ll} \phi$ consists first in constructing the NBA $\Aa$ that captures $\neg \phi$, and then verifying that $\Sys \not\models_{\Ll} \neg \phi$ by demonstrating that no execution trace of the system is accepted by $\Aa$.

To verify whether a system satisfies an LTL property, we use a closure certificate on the product $\Sys \otimes \Aa$ as follows.
\begin{definition}[CC for LTL]
\label{def:cc_ltl}
   Consider a system $\Sys = (\Xx, \Xx_0, f)$ and an NBA $\Aa= (\Sigma, \Qq, \Qq_0, \delta, \Qq_F)$ representing the complement of an LTL formula $\phi$. 
   A function $\Tt: \Xx \times \Qq \times  \Xx \times \Qq \to \R$ is a CC for $\Sys \otimes \Aa$ if $\exists \eta \in \R_{> 0}, \lambda \in \R_{\geq 0}$ such that $\forall x, y, y' \in \Xx$, $\forall x_0 \in \Xx_0$, $\forall x' \in f(x)$, $\forall n, p \in \Qq$, $\forall n' \in \delta(n, \Ll(x))$, $\forall q_0 \in \Qq_0$ and $\forall r, r' \in \Qq_F$:
    \begin{align}
        \label{eq:cc_ltl1}
        & \Tt\big( (x, n), (x', n') \big) \geq 0, \\
        \label{eq:cc_ltl2} 
        & \Tt \big((x, n), (y, p) \big) \geq \lambda \Tt \big( (x', n'), (y, p) \big), \\
        \label{eq:cc_ltl3}
        & \Big( \Tt\big((x_0,q_0),(y,r) \big) \geq 0 \Big) \wedge \Big( \Tt( (y,r), (y',r') ) \geq 0 \Big) \implies \nonumber \\
        & \qquad \Big(\Tt \big( (x_0,q_0),(y',r') \big) \leq \Tt \big( (x_0,q_0),(y,r) \big) - \eta \Big).  
    \end{align}
\end{definition}
\begin{theorem}[CCs verify LTL~\cite{murali2024closure}]
\label{thm:cc_ltl}
   Consider a system $\Sys = (\Xx, \Xx_0, f)$ and an LTL formula $\phi$. Let NBA $\Aa$ represent the complement of the specification, \textit{i.e}, $\neg \phi$. The existence of a CC satisfying conditions~\eqref{eq:cc_ltl1}-\eqref{eq:cc_ltl3} implies that $\Sys \models_{\Ll} \phi$. 
\end{theorem}

In the next section, we define and describe the utility of vector co-B\"uchi ranking functions and vector closure certificates.

\section{Vector Certificates for $\omega$-regular Specifications}
\label{sec:vcs}
We first mention the vector formulation of ranking functions dubbed vector co-B\"uchi ranking function (VCBRF) to argue for finite visits.
\begin{definition}[VCBRF for Persistence]
\label{def:vcbrf} 
A vector of functions $\B$ consisting of $\Bb_i: \Xx \rightarrow \R$,  $\forall  i \in \{1,\ldots, k \}$, is a VCBRF for a system $\Sys$ with respect to a set of states $\Xx_{VF} \subseteq \Xx$ that must be visited only finitely often if $\exists \eta \in \R_{>0}$ and $\exists A_1 \in \R_{\geq 0}^{k\times k}$ such that $\forall x \in \Xx$, $\forall x_0 \in \Xx_0$, $\forall x' \in f(x)$, $\forall y \in \Xx \backslash \Xx_{VF}$, $\forall y' \in f(y)$, $\forall z \in \Xx_{VF}$, $\forall z' \in f(z)$:
\begin{align}
    \label{eq:vcbrf_1}
    &\B (x_0) \geq 0, &&
    \\
    \label{eq:vcbrf_2}
    &\B(x') \geq A_1 \B(x), 
    \\
    \label{eq:vcbrf_3}
    &(\B(y) \geq 0) \implies (\B(y) - \B(y') \geq 0), \text{ and}
    \\
    \label{eq:vcbrf_4}
    &(\B(z) \geq 0) \implies (\B(z) - \B(z') - \eta \1_k \geq 0),
\end{align}
where the inequalities are element-wise.
\end{definition}
\begin{theorem}[VCBRFs imply persistence]
\label{thm:vcbrf}
The existence of a VCBRF $\Bb_i: \Xx \rightarrow \R$, $\forall i \in \{1,\ldots,k \}$, satisfying conditions~\eqref{eq:vcbrf_1}-\eqref{eq:vcbrf_4} for a system $\Sys = (\Xx, \Xx_0, f)$ implies that the state sequences of the system visit the set $\Xx_{VF}$ finitely often.
\end{theorem}
\begin{proof}
Conditions~\eqref{eq:vcbrf_1} and~\eqref{eq:vcbrf_2} serve to overapproximate the reachable states of the system (similar to VBCs) and, inductively, we have $\B(x) \geq 0$ over all reachable states $x \in \Xx$. 
Consequently, from~\eqref{eq:vcbrf_3} and~\eqref{eq:vcbrf_4}, we have $\B(y) - \B(y') \geq 0$ for reachable $y \in \Xx \backslash \Xx_{VF}$ and $\B(z) - \B(z') - \eta \1_k \geq 0$ for reachable $z \in \Xx_{VF}$, respectively.
This means that as the system evolves, the functions do not increase for $x \in \Xx \backslash \Xx_{VF}$ and are strictly decreasing for $x \in \Xx_{VF}$. 
In order to maintain a lower bound of $0$ for all functions over the reachable states, the functions cannot be decreasing infinitely as we continue visiting $\Xx_{VF}$.  
Thus, the existence of a vector of functions satisfying conditions~\eqref{eq:vcbrf_1}-\eqref{eq:vcbrf_4} serves as a proof that $\Xx_{VF}$ is visited finitely often.
\end{proof}

Remark that in condition~\eqref{eq:vcbrf_2}, there is a nonnegative matrix $A_1$ relating $\B(x')$ and $\B(x)$. 
However, such a matrix cannot be incorporated under conditions~\eqref{eq:vcbrf_3} and~\eqref{eq:vcbrf_4}, since it would no longer guarantee a nonincreasing or, respectively, strictly decreasing behavior. 
Employing a nonidentity matrix may breach these conditions and, as a result, could permit visiting $\Xx_{VF}$ infinitely many times.

To verify whether a given system satisfies a desired LTL property, we use VCBRFs on the product $\Sys \otimes \Aa$ as follows.
\begin{definition}[VCBRF for LTL]
\label{def:vcbrf_ltl}
Consider a system $\Sys = (\Xx, \Xx_0, f)$ and an NBA $\Aa= (\Sigma, \Qq, \Qq_0, \delta, \Qq_F)$ representing the complement of an LTL formula $\phi$.
A vector of functions $\B$ consisting of $\Bb_i: \Xx \times \Qq \rightarrow \R$,  $\forall i \in \{1,\ldots, k \}$, is a VCBRF for $\Sys \otimes \Aa$ if $\exists \eta \in \R_{>0}$ and $\exists A_1,A_2,A_3 \in \R_{\geq 0}^{k\times k}$ such that $\forall x \in \Xx$, $\forall x_0 \in \Xx_0$, $\forall x' \in f(x)$ and $\forall n \in \Qq$, $\forall q_0 \in \Qq_0$, $\forall q \in \Qq \backslash \Qq_F$, $\forall r \in \Qq_F$, $\forall n' \in \delta(n, \Ll(x))$, $\forall q' \in \delta(q, \Ll(x))$ and $\forall r' \in \delta(r, \Ll(x))$:
\begin{align}
    \label{eq:vcbrf_ltl1}
    &\B((x_0,q_0)) \geq 0,
    \\
    \label{eq:vcbrf_ltl2}
    &\B((x',n')) \geq A_1 \B((x,n)), 
    \\
    \label{eq:vcbrf_ltl3}
    &(\B((x,q)) \geq 0) \implies (\B((x,q)) - \B((x',q')) \geq 0),
    \\
    \label{eq:vcbrf_ltl4}
    &(\B((x,r)) \geq 0) \implies (\B((x,r)) - \B((x',r')) - \eta \1_k \geq 0),
\end{align}
where the inequalities are element-wise.
\end{definition}
\begin{theorem}[VCBRFs verify LTL]
\label{thm:vcbrf_ltl}
Consider a system $\Sys = (\Xx, \Xx_0, f)$ and an LTL formula $\phi$. Let NBA $\Aa$ represent the complement of the specification, \textit{i.e}, $\neg \phi$.
The existence of a VCBRF $\Bb_i: \Xx \times \Qq \rightarrow \R$, $\forall i \in \{1,\ldots,k \}$, satisfying conditions~\eqref{eq:vcbrf_ltl1}-\eqref{eq:vcbrf_ltl4} over $\Sys \otimes \Aa$ implies $\Sys \models_{\Ll} \phi$.
\end{theorem}
\begin{proof}
    Observe that a VCBRF that satisfies conditions~\eqref{eq:vcbrf_ltl1}-\eqref{eq:vcbrf_ltl4} is a VCBRF for persistence over the product $\Sys \otimes \Aa$.
    From Theorem~\ref{thm:vcbrf}, we observe that the product system visits the accepting states only finitely often. 
    Thus, we infer that no trace of the system is present in the language of the NBA $\Aa$, thus verifying the LTL specification.
\end{proof}
We now introduce a notion of vector closure certificates (VCCs) and define them for safety, persistence, and general $\omega$-regular specifications.

\begin{definition}[VCC for Safety]
\label{def:vcc_safe}
   Consider a system $\Sys = (\Xx, \Xx_0, f)$. A vector of $k$ functions $\T$ consisting of $\Tt_i: \Xx \times \Xx \to \R,\forall i \in \set{1,\ldots,k}$ is a VCC for $\Sys$ with respect to a set of unsafe states $\Xx_{u} \subseteq \Xx$ if $\exists \eta \in \R_{>0}$ and $A \in \R_{\geq 0}^{k\times k}$ such that $\forall x, y \in \Xx$, $\forall x_0 \in \Xx_0$, $\forall x_u \in \Xx_u$, $\forall x' \in f(x)$:
    \begin{align}
        \label{eq:vcc_safe1}
        & \T(x,x') \geq 0, \\
        \label{eq:vcc_safe2}
        &\T(x,y) \geq A\T(x', y), \\
        \label{eq:vcc_safe3}
        & \bigvee_{i=1}^{k} \big(\Tt_i(x_0, x_u) \leq - \eta \big),
    \end{align}
    where the inequalities are element-wise.
\end{definition}
We now describe the utility of VCCs for safety.
\begin{theorem}[VCCs imply Safety]
    \label{thm:vcc_safe}
    Consider a system $\Sys = (\Xx,\Xx_0,f)$ and a set of unsafe states $\Xx_u$. The existence of a VCC that satisfies the conditions~\eqref{eq:vcc_safe1}-\eqref{eq:vcc_safe3} implies its safety.
\end{theorem}
\begin{proof}
Let $j \in \set{1,\ldots,k}$ be an index such that $\Tt_j(x_0, x_u) \leq - \eta$ (condition~\eqref{eq:vcc_safe3}).  Moreover, from condition~\eqref{eq:vcc_safe1}, $\Tt_j(x,x') \geq 0$, and inductively, we get $\Tt_j(x',y) \geq 0$. Since $A$ is nonnegative, $A\T(x',y) \geq 0$. 
Then, if $\T(x,y) \geq A\T(x',y)$, $\T(x,y) \geq 0$ as in condition~\eqref{eq:vcc_safe2}. From this we get $\Tt_j(x,y) \geq \sum_{i = 1}^{k} A_{ji}\Tt_i(x',y) \geq 0$.
These three conditions are the conditions of CC for safety over $\Tt_j(\cdot,\cdot)$ according to Definition~\ref{def:cc_safe}. 
Thus, the proof of safety follows accordingly.
\end{proof}

\begin{definition}[VCC for Persistence]
\label{def:vcc_pers}
   Consider a system $\Sys = (\Xx, \Xx_0, f)$. A vector of $k$ functions $\T$ consisting of $\Tt_i: \Xx \times \Xx \to \R,\forall i \in \set{1,\ldots,k}$ is a VCC for $\Sys$ with a set of states $\Xx_{VF} \subseteq \Xx$ that must be visited finitely often if $\exists \eta \in \R_{>0}$ and $A \in \R_{\geq 0}^{k\times k}$ such that $\forall x, y\in \Xx$, $\forall x_0 \in \Xx_0$, $\forall x' \in f(x)$, $\forall y', y''\in \Xx_{VF}$:
    \begin{align}
        \label{eq:vcc_pers1}
        & \T(x,x') \geq 0, \\
        \label{eq:vcc_pers2}
        &\T(x,y) \geq A\T(x', y), \\
        \label{eq:vcc_pers3}
        & \big(\T (x_0,y') \geq 0 \big) \wedge \big( \T(y',y'') \geq 0 \big) \implies \nonumber \\
        & \qquad \bigvee_{i=1}^{k} \big(\Tt_i(x_0,y'') \leq \Tt_i(x_0,y') - \eta \big),
    \end{align}
    where the inequalities are element-wise.
\end{definition}

VCCs can be used to verify persistence, as formally stated below.
\begin{theorem}[VCCs imply Persistence]
    \label{thm:vcc_pers}
    Consider a system $\Sys = (\Xx,\Xx_0,f)$. The existence of a VCC satisfying conditions~\eqref{eq:vcc_pers1}-\eqref{eq:vcc_pers3} implies that the state sequences of the system visit the set $\Xx_{VF}$ only finitely often.
\end{theorem}
\begin{proof}
From conditions~\eqref{eq:vcc_pers1} and~\eqref{eq:vcc_pers2}, $\Tt_d(x,x') \geq 0$ and $\Tt_d(x,y) \geq \sum_{i = 1}^{k} A_{di}\Tt_i(x',y) \geq 0$, respectively, for all $d \in \{1,\ldots,k\}$. 
Let $j \in \{1,\ldots,k\}$ be an index such that $\big(\T (x_0,y') \geq 0 \big) \wedge \big( \T(y',y'') \geq 0 \big) \implies 
\big(\Tt_j(x_0,y'') \leq \Tt_j(x_0,y') - \eta \big)$ (condition~\eqref{eq:vcc_pers3}). 
This implies that condition $\big[\land_{i \neq j} (\Tt_i(x_0,y') \geq 0) \wedge (\Tt_i(y',y'') \geq 0)\big] \wedge \big(\Tt_j(x_0,y') \geq 0 \big) \wedge \big(\Tt_j(y',y'') \geq 0 \big) \implies 
\big(\Tt_j(x_0,y'') \leq \Tt_j(x_0,y') - \eta \big)$ is satisfied.
Now, consider $\Xx_{VF} = \cup_{i = 1}^{m} \Xx_{VF_i}$. 
We will prove that $\Xx_{VF}$ is visited finitely often by contradiction. 
Take a state sequence $s = (x_0,x_1,\ldots) \in \Xx^{\omega}$ and assume that it visits $\Xx_{VF}$ infinitely often. 
Since $\Tt_j(x,x') \geq 0$, $\Tt_j(x_i,x_{i+1}) \geq 0$ for every $x_i \in s$. 
Similarly, from $\Tt_j(x,y) \geq \sum_{i = 1}^{k} A_{ji}\Tt_i(x',y) \geq 0$, $\Tt_j(x_0,x_i) \geq 0$, and $\Tt_j(x_i,x_{\ell}) \geq 0$ for every $i \in \N$ and every $\ell \geq i + 1$.
Let $(y_0,y_1,\ldots)$ be a subsequence of $s$ that visits $\Xx_{VF}$ only finitely often (i.e., $s = (x_0,x_1,\ldots,y_0,\ldots,y_1,\ldots)$). 
By Ramsey's theorem~\cite{ramsey1987problem}, there exists a subsequence $(z_0,z_1,\ldots) \in \Xx_{VF_i}$ of $s$ that visits $\Xx_{VF_i}$ infinitely often for some $i \in \{1,\dots,m\}$.
From the previous conclusions, we derive $\Tt_j(x_0,z_i) \geq 0$ and $\Tt_j(z_i,z_{\ell}) \geq 0$ $\forall i \in \N$ and $\forall \ell \geq i + 1$. 
Thus, we know that $\Tt_j(x_0,z_\ell)$ is bounded below by 0. 
From condition~\eqref{eq:vcc_pers3} and by induction, we have $\Tt_j(x_0,z_\ell) \leq \Tt_j(x_0,z_0) - \ell \eta, \forall \ell \in \N_{\geq 1}$. 
Therefore, $\exists \ell \in \N$ such that $\Tt_j(x_0,z_\ell) \leq \Tt_j(x_0,z_0) - \ell \eta < 0$, which is a contradiction. Thus, $\Xx_{VF}$ is visited finitely often.
\end{proof}

Definition~\ref{def:vcc_pers} can be extended to LTL specifications on the product of the system and an NBA as follows.
\begin{definition}[VCC for LTL]
\label{def:vcc_ltl}
   Consider a system $\Sys = (\Xx, \Xx_0, f)$ and an NBA $\Aa= (\Sigma,\Qq, \Qq_0, \delta, \Qq_F)$ representing the complement of an LTL formula $\phi$. A vector of $k$ functions $\T$ consisting of $\Tt_i: \Xx \times \Qq \times \Xx \times \Qq \to \R,\forall i \in \set{1,\ldots,k}$, is a VCC for $\Sys \otimes \Aa$ if $\exists \eta \in \R_{ > 0}$ and $A \in \R_{\geq 0}^{k\times k}$ such that $\forall x, y, y' \in \Xx$, $\forall x_0 \in \Xx_0$, $\forall x' \in f(x)$, $\forall n,p \in \Qq$, $\forall n' \in \delta(n, \Ll(x))$,  $\forall q_0 \in \Qq_0$, and $\forall r,r' \in \Qq_F$:
    \begin{align}
        \label{eq:vcc_ltl1}
        & \T\big( (x, n), (x', n') \big) \geq 0, \\
        \label{eq:vcc_ltl2}
        & \T \big((x, n), (y, p) \big) \geq A \T \big( (x', n'), (y, p) \big), \\
        \label{eq:vcc_ltl3}
        & \Big(\T \big((x_0,q_0),(y,r) \big) \geq 0 \Big) \wedge \Big( \T( (y,r), (y',r') ) \geq 0 \Big) \implies \nonumber \\
        & \quad 
        \bigvee_{i=1}^{k} \Big(\Tt_i \big( (x_0,q_0),(y',r') \big) \leq \Tt_i \big( (x_0,q_0),(y,r) \big) - \eta \Big),
    \end{align}
    where the inequalities are element-wise.
\end{definition}
\begin{theorem}[VCCs verify LTL]
\label{thm:vcc_ltl}
   Consider a system $\Sys = (\Xx, \Xx_0, f)$ and an LTL formula $\phi$. Let NBA $\Aa$ represent the complement of the specification, \textit{i.e}, $\neg \phi$. The existence of a VCC satisfying conditions~\eqref{eq:vcc_ltl1}-\eqref{eq:vcc_ltl3} implies that $\Sys \models_{\Ll} \phi$. 
\end{theorem}
\begin{proof}
    Observe that a VCC satisfying the conditions~\eqref{eq:vcc_ltl1}-~\eqref{eq:vcc_ltl3} mirrors a VCC for persistence of the product $\Sys \otimes \Aa$.
    From Theorem~\ref{thm:vcc_pers}, we observe that the product system visits the accepting states only finitely often. Thus, we infer that no trace of the system is present in the language of the NBA $\Aa$, thus verifying the LTL specification.
\end{proof}

The benefit of the VCC formulation over CC's for safety (and persistence) can be described as follows. 
Consider $\Xx_{u} = \bigcup_{j=1}^{m} \Xx_{u_j}$ (resp. $\Xx_{VF} = \bigcup_{j=1}^{m} \Xx_{VF_j}$). 
Then, following the disjunctive well-foundedness of closure certificates, without loss of generality, one can choose different $j$ indices for condition~\eqref{eq:vcc_safe3}  (resp. condition~\eqref{eq:vcc_pers3}) while modifying $\eta$ to $\eta_j$, such that the choice of $\eta_j$ for each region $\Xx_{u_j}$ (resp. $\Xx_{VF_j}$) is different.
Observe that for VCBRFs, the decrease is also imposed across all functions, and we cannot dedicate each function $\Bb_j$ to $\Xx_{VF_j}$.
In addition, note that conditions~\eqref{eq:cc_safe1}-\eqref{eq:cc_safe2} provide the transition invariant in the form of a single function, while conditions~\eqref{eq:vcc_safe1}-\eqref{eq:vcc_safe2} provide the transition invariant as linear combinations of the different functions.
Hence, this flexibility in choosing a different function and factor $\eta$, as well as the generality of the transition invariant, makes the formulation of VCCs much more powerful than the standard formulation of CC. 
Similarly, for the LTL specification, one can choose different $j$ indices for each accepting state $r$ of NBA $\Aa$ and use different $\eta_j$ instead of imposing condition~\eqref{eq:vcc_ltl3} on all possible pairs of accepting states $r$ and $r'$.
As a result, the use of multiple functions in this vector form allows us to utilize simpler templates of functions as certificates. Formally, we can formulate this deductive power of VCCs in a manner similar to that of VBCs \cite{sogokon2018vector}.

\begin{theorem}
    Every polynomial CC is (trivially) a VCC. 
    The converse is not necessarily true; i.e., there exist polynomial VCCs sufficient for proving certain $\omega$-regular properties where a polynomial CC of the same template does not exist. 
\end{theorem}

\begin{proof}
 Interestingly, setting the $A$ matrix as the identity reduces VCCs to CCs. 
    See Section~\ref{sec:casestudies} for examples of VCCs with polynomial degrees lower than scalar CCs.
    We will demonstrate the existence of polynomial VCC when a polynomial CC of the same template does not exist in the case of safety. Take the system shown in Figure~\ref{fig:counterexample}. Assume there exists a quadratic CC and consider the conditions shown in Definition~\ref{def:cc_safe}. Based on the following sets of inequalities, it can be verified that there is no feasible quadratic template for $\Tt(x,y)$. We used Gurobi~\cite{gurobi} to confirm the infeasibility of this bilinear problem (solved as a Mixed-Integer Quadratically Constrained Program (MIQCP)). Note that the state variable is one-dimensional.
    {\allowdisplaybreaks
    \begin{align*}
        \Tt(0,2) &\geq 0,\\
        \Tt(2,4) &\geq 0,\\
        \Tt(0,1) &\leq -\eta,\\
        \Tt(0,3) &\leq -\eta,\\
        \Tt(0,0) &\geq \lambda \Tt(2,0),\\
        \Tt(0,4) &\geq \lambda \Tt(2,4),\\
        \lambda &\geq 0,
    \end{align*}
    }
    where $\eta$ was as low as $10^{-4}$.
    However, we were able to obtain the following quadratic VCC with two components, $\eta = 10^{-3}$ and $A = \big(
    \begin{smallmatrix}
      4.703 & 4.703\\
      1.621 & 0.354
    \end{smallmatrix}\big)$, according to Definition~\ref{def:vcc_safe}:
    \begin{align*}
    \Tt_1(x,y) &= 847.87xy - 883.63y^2 - 3391.47x + 4442.96y\\
    & - 3633.71,\\
    \Tt_2(x,y) &= 0.080xy + 0.081y^2 - 0.319x - 0.564y + 0.965.
    \end{align*}
\end{proof}

\begin{figure}[t]
    \centering
    \begin{tikzpicture}[node distance =2cm]
    \node[initial, state, draw, initial text =, fill=green!20!white] (0) at (-3,0) {$q_0$};
    \node[state, fill = red!20!white] (1) at (-1.5,0) {$q_1$};
    \node[state, fill = blue!10!white] (2) at (0,0) {$q_2$};
    \node[state, fill = red!20!white] (3) at (1.5, 0) {$q_3$};
    \node[state, fill = blue!10!white] (4) at (3, 0) {$q_4$};
    \path[->]
    (0) edge[bend left = 6em] node[above]{} (2)
    (2) edge[bend left = 6em] node[above]{} (4)
    (4) edge[loop right] node[above]{} (4)
    (2) edge[bend left = 6em] node[above]{} (0)
    (0) edge[loop above] node[above]{} (0)
    (1) edge[] node[above]{} (2)
    (3) edge[] node[above]{} (4);   
    \end{tikzpicture}
    \caption{Illustrative example demonstrating the existence of VCCs over barrier certificates. The initial state is shown in green. The unsafe states are shown in red.}
    \label{fig:counterexample}
\end{figure}
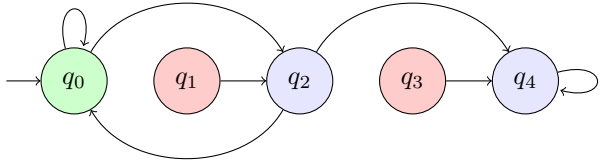

Observe that in the formulation of both VCBRFs and VCCs, setting $k = 1$ retrieves the standard conditions for the ranking functions and the closure certificates, respectively.
In the next section, we introduce a computational tool for finding VCBRFs and VCCs and discuss the benefits of these certificates from the perspective of computational complexity.

\section{Computation of Vector Certificates}
\label{sec:synth}
This section presents an approach to synthesize VCBRFs and VCCs using sum-of-squares (SOS)~\cite{parrilo2003semidefinite} programming when the relevant sets of the system $\Sys = (\Xx, \Xx_0, f)$ are semi-algebraic sets and the transition function $f:\Xx \to \Xx$ is polynomial.
A set $Y \subseteq \R^n$ is semi-algebraic if it can be defined with a vector of polynomial inequalities $h(x)$ as $Y = \{ x \mid h(x) \geq 0 \}$, where the inequalities are element-wise.
Moreover, to use SOS programming, we assume that the nonnegative matrix $A$ is provided.
Lastly, to find VCCs, we first fix the template to be a linear combination of user-defined basis functions:
$
   \Tt(x,y) = \mathbf{c}^T\mathbf{p}(x,y) = \sum_{m = 1}^{n} c_m p_m(x,y), 
$
where functions $p_m$ are monomials over state variables $x$ and $y$, and $c_1, \ldots, c_n$ are real coefficients.
For VCBRFs, the monomials $\mathbf{p}(x)$ are only over the state variable $x$.

We now describe the relevant SOS conditions for safety, persistence, and LTL specifications.

\subsection{Safety}
\label{subsec:find_safety}
To adopt an SOS approach to find VCCs as in Definition~\ref{def:vcc_safe}, we consider $\Xx_{u} = \cup_{j=1}^{m} \Xx_{u_j}$ and the sets $\Xx$, $\Xx_0$, and $\Xx_{u_j}$ to be semi-algebraic sets defined with the help of vectors of polynomial inequalities $g_{A}(x)$, $g_{0}(x)$, and $g_{u,j}(x)$, respectively.
As these sets are semi-algebraic, the sets $\Xx \times \Xx$ and $\Xx_0 \times \Xx_{u_j}$ are semi-algebraic
with the corresponding vectors $g_{B}$ and $g_{C,j}$, respectively.
Then the search for a VCC, as in Definition~\ref{def:vcc_safe}, reduces to showing that the following polynomials are SOS:
\begin{align}
    \label{eq:sos_vcc_safe1} 
    & \Tt_i(x, x') \!-\! \lambda_{A,i}^T(x)g_{A}(x),~~ \forall i \in \{1,\ldots,k\}, \\
    \label{eq:sos_vcc_safe2}
    &\Tt_{i}(x, y) \!-\! \sum_{j = 1}^{k} A_{ij} \Tt_j(x', y)-\lambda_{B,i}^T(x,y)g_{B}(x,y),\\ 
    &\hspace{10em} \forall i \in \{1,\ldots,k\},\\
    \label{eq:sos_vcc_safe3}
    &-\eta_j - \Tt_j(x_0, x_u) - \lambda_{C,j}^T(x_0,x_u)g_{C,j}(x_0,x_u),\\ 
    &\hspace{10em} \forall j \in \{1,\ldots,m \},
\end{align}
where $x' = f(x)$, $\eta_j \in\R_{ > 0}$, $m \leq k$, $A \in \R_{\geq 0}^{k\times k}$ and multipliers $\lambda_{A,i}$, $\lambda_{B,i}$, $\lambda_{C,j}$, are SOS polynomials over the state variable $x$, the state variables $x,y$, and the state variables $x_0,x_u$ over the sets $\Xx$, $\Xx \times \Xx$, and $\Xx_0 \times \Xx_{u_j}$, respectively.

\subsection{Persistence}
\label{subsec:find_persistence}

Here, we will cover the conditions for both VCBRFs and VCC.
To search for VCBRFs as in Definition~\ref{def:vcbrf}, we consider the sets $\Xx$, $\Xx_0$, and $\Xx_{VF}$ to be semi-algebraic sets defined with the help of vectors of polynomial inequalities $g_{A}(x)$, $g_{0}(x)$, and $g_{B}(x)$, respectively.
As $\Xx_{VF}$ is semi-algebraic, $\Xx \setminus \Xx_{VF}$ is also semi-algebraic with corresponding vector $g_{C}(x)$.
Then the search for a VCBRF, as in Definition~\ref{def:vcbrf}, reduces to showing that the following polynomials are SOS $\forall i \in \{1,\ldots,k\}$:
\begin{align}
    \label{eq:sos_vcbrf_pers1} 
    & \Bb_i(x_0) - \lambda_{0,i}^T(x)g_{0}(x),\\
    \label{eq:sos_vcbrf_pers2}
    &\Bb_{i}(x') - \sum_{j = 1}^{k} A_{1_{ij}} \Bb_j(x) - \lambda_{A,i}^T(x)g_{A}(x),\\
    \label{eq:sos_vcbrf_pers3}
    & \Bb_{i}(y) - \Bb_{i}(y') - \sum_{j = 1}^{k} A_{2_{ij}} \Bb_j(y) - \lambda_{C,i}^T(y)g_{C}(y),\\
    \label{eq:sos_vcbrf_pers4}
    & \Bb_{i}(z) - \Bb_{i}(z') - \eta - \sum_{j = 1}^{k} A_{3_{ij}} \Bb_j(z) - \lambda_{B,i}^T(z)g_{B}(z),
\end{align}
where $x' = f(x)$, $y' = f(y)$, $z' = f(z)$, $\eta \in\R_{ > 0}$, $A_1, A_2, A_3 \in \R_{\geq 0}^{k\times k}$ and  the multipliers $\lambda_{A,i}$, $\lambda_{B,i}$, $\lambda_{C,j}$, are SOS polynomials over the state variable $x$, $z$ and $y$ over the sets $\Xx$, $\Xx_{VF}$ and $\Xx \setminus \Xx_{VF}$, respectively.

To find VCCs as in Definition~\ref{def:vcc_pers}, we consider $\Xx_{VF} = \cup_{j=1}^{m}\Xx_{VF_j}$ and the sets $\Xx$, $\Xx_0$, and $\Xx_{VF_j}$ to be semi-algebraic sets defined with the help of vectors of polynomial inequalities $g_{A}(x)$, $g_{0}(x)$, and $g_{VF,j}(x)$, respectively.
As these sets are semi-algebraic, the sets $\Xx \times \Xx$ and $\Xx_0 \times \Xx_{VF_j} \times \Xx_{VF_j}$ are semi-algebraic
with the corresponding vectors $g_{B}$ and $g_{C,j}$, respectively.
Then the search for a VCC, as in Definition~\ref{def:vcc_pers}, reduces to showing that the following polynomials are SOS:
\begin{align}
    \label{eq:sos_vcc_pers1} 
    & \Tt_i(x, x') - \lambda_{A,i}^T(x)g_{A}(x),~~ \forall i \in \{1,\ldots,k\},\\
    \label{eq:sos_vcc_pers2}
    &\Tt_{i}(x, y) - \sum_{j = 1}^{k} A_{ij} \Tt_j(x', y)-\lambda_{B,i}^T(x,y)g_{B}(x,y),\nonumber\\
    &\hspace{10em} \forall i \in \{1,\ldots,k\},\\
    \label{eq:sos_vcc_pers3}
    & \Tt_j(x, y) - \eta_j - \Tt_j(x,y') - \sum_{i = 1}^{k} \big(\gamma_i \Tt_i(x,y) + \rho_i \Tt_i(y,y')\big) \nonumber \\
    & \quad - \lambda_{C,j}^T(x_0,x_u)g_{C,j}(x_0,x_u), ~~ \forall j \in \{1,\ldots,m \},
\end{align}
where $x' = f(x)$,  $\eta_j \in\R_{ > 0}$, $\gamma_i, \rho_i \in \R_{ \geq 0}$, $m \leq k$, $A \in \R_{\geq 0}^{k\times k}$ and  the multipliers $\lambda_{A,i}$, $\lambda_{B,i}$, $\lambda_{C,j}$, are SOS polynomials over the state variable $x$, the state variables $x,y$, and the state variables $x,y,y'$ over the sets $\Xx$, $\Xx \times \Xx$, and $\Xx_0 \times \Xx_{VF_j} \times \Xx_{VF_j}$, respectively.

\subsection{LTL Specifications}
\label{subsec:find_ltl}

For a given LTL specification, the relevant NBA has finitely many letters $\sigma \in \Sigma$. Without loss of generality, the set $\Xx$ can be partitioned into finitely many partitions $\Xx_{\sigma_1}, \ldots, \Xx_{\sigma_p}$, where $\forall x \in \Xx_{\sigma_m}$ we have $\Ll(x) = \sigma_m$.
Given an element $\sigma_m \in \Sigma$, we can uniquely characterize the relation $\delta_{\sigma_i}$ as $(q'_i, q_i) \in \delta_{\sigma_i}$ if and only if  $q'_i \in \delta(q_i, \sigma_i)$.
We assume that the sets $\Xx$, $\Xx_{0}$, and $\Xx_{\sigma_m}$ $\forall \sigma_m$ are semi-algebraic and characterized by polynomial vectors of inequalities $g(x)$, $ g_0(x)$, and $g_{\sigma_m, A}(x)$, respectively.
Similarly, we consider polynomial vectors of inequalities $g_{\sigma_m, B}(x, y)$ in the product space $\Xx_{\sigma_m} \times \Xx$ and $g_{C,j}(x_0, y, y')$ in $\Xx_0 \times \Xx \times \Xx$.

The search for VCBRF can be reduced to showing that $\forall x \in \Xx$, $\forall x_0 \in \Xx_0$, and $\forall n \in \Qq$, $\forall q_0 \in \Qq_0$, $\forall q \in \Qq \backslash \Qq_F$, $\forall r \in \Qq_F$, $\forall n' \in \delta(n, \Ll(x))$, $\forall q' \in \delta(q, \Ll(x))$ and $\forall r' \in \delta(r, \Ll(x))$, the following polynomials are SOS $\forall i \in \{1,\ldots,k\}$:
\begin{align}
    \label{eq:sos_vcbrf_ltl1} 
    & \Bb_i^{(q_0)}(x_0) - \lambda_{0,i}^T(x_0)g_{0}(x_0),\\
    \label{eq:sos_vcbrf_ltl2}
    &\Bb_{i}^{(n')}(x') - \sum_{j = 1}^{k} A_{1_{ij}} \Bb_j^{(n)}(x) - \lambda_{\sigma_m, A, i}^T(x)g_{\sigma_m, A}(x),\\
    \label{eq:sos_vcbrf_ltl3}
    & \Bb_{i}^{(q)}(x) - \Bb_{i}^{(q')}(x') - \sum_{j = 1}^{k} A_{2_{ij}} \Bb_j^{(q)}(x) \nonumber\\
    &\quad - \bar{\lambda}_{\sigma_m, A, i}^T(x)g_{\sigma_m, A}(x),\\
    \label{eq:sos_vcbrf_ltl4}
    & \Bb_{i}^{(r)}(x) - \Bb_{i}^{(r')}(x') - \eta - \sum_{j = 1}^{k} A_{3_{ij}} \Bb_j^{(r)}(x) \nonumber\\
    &\quad - \hat{\lambda}_{\sigma_m, A, i}^T(x)g_{\sigma_m, A}(x),
\end{align}
where $x' = f(x)$, $A_1, A_2, A_3 \in \R_{\geq 0}^{k\times k}$, $\eta\in \R_{> 0}$, and $\lambda_{0, i}$, $\lambda_{\sigma_m, A, i}$, $\bar{\lambda}_{\sigma_m, A, i}$, and $\hat{\lambda}_{\sigma_m, A, i}$ are SOS polynomials of appropriate dimensions over $\Xx_{0}$ and $\Xx_{\sigma_m}$. 
Remark that the superscripts $(u)$ for each function $\Bb_i$ above denote automata states.

Now, we can reduce the search for a VCC to showing that $\forall x, y, y' \in \Xx$, $\forall n, p \in \Qq$, $\forall q_0 \in \Qq_0$, $\forall r \in {\Qq_F}$, and $\forall \sigma_m \in \Sigma$, such that $n' \in \delta_{\sigma_m} (n)$, the following polynomials are SOS:
\begin{align}
\label{eq:sos_vcc_ltl1}
&\Tt_i^{(n,n')}(x, x') - \lambda^T_{\sigma_m, A, i}(x)g_{\sigma_m, A}(x),~~ \forall i \in \{1,\ldots,k\}, \\
\label{eq:sos_vcc_ltl2}
& \Tt_{i}^{(n,p)}(x, y) - \sum_{j = 1}^{k} A_{ij} \Tt_j^{(n',p)}(x',y) \nonumber \\
& \quad -\lambda^T_{\sigma_m, B, i}(x, y) g_{\sigma_m, B}(x, y),~~ \forall i \in \{1,\ldots,k\},\\ 
\label{eq:sos_vcc_ltl3}
& \Tt_j^{(q_0,r)} (x_0,y) - \eta_j - \sum_{i = 1}^{k} \big(\gamma_i \Tt_i^{(q_0,r)}(x_0,y) + \rho_i \Tt_i^{(r,r)}(y,y')\big) \nonumber\\
&\quad - \Tt_j^{(q_0,r)}(x_0,y') - \lambda^T_{C, j}(x_0, y, y') g_{C,j}(x_0, y, y'), \nonumber\\
&\hspace{10em} \forall j \in \{1,\ldots,|\Qq_F|\},
\end{align}
where $x' = f(x)$, $\lambda_{\sigma_m, A, i}$, $\lambda_{\sigma_m, B, i}$, and $\lambda_{\sigma_m, C, j}$ are SOS polynomials of appropriate dimensions over their respective regions, $ A \in \R_{\geq 0}^{k\times k}$, $|\Qq_F| \leq k$, $\eta_j \in \R_{> 0}$, and $\gamma_i, \rho_i \in \R_{\geq 0}$. Note that the $(u,v)$ superscripts for each function $\Tt_i$ above denote automata state pairs.

A common tool to search for SOS polynomial certificates is to use solvers such as~\cite{wang2021tssos}.
The complexity of determining whether the above equations are SOS is $O\big(k \binom{2n+d}{d}^2 \big)$~\cite{parrilo2003semidefinite}, when searching for VCCs for safety or persistence, where $n$ is the dimension of the set of states and $2d$ is the degree of the polynomial.
The complexity of verifying LTL specifications is $O\big( k|\Qq|^2 \binom{2n+d}{d}^2\big)$, where $|\Qq|$ indicates the number of automata states~\cite{murali2024closure}.
By allowing lower degrees of functions to act as certificates in the vector formulation, we incur a linear complexity in $k$, while alleviating the polynomial complexity due to the degree $2d$.
Thus, one may use existing techniques to reduce the computational burden in the search for such certificates.
For VCBRFs, the term $\binom{2n+d}{d}$ is replaced by $\binom{n+d}{d}$, making the computation for VCBRFs relatively easier and faster, especially for systems of higher dimensions.

\section{Case Studies}
\label{sec:casestudies}

This section covers case studies used to demonstrate the utility of VCBRFs and VCCs for safety, persistence, and LTL specifications, where appropriate. 
Given a system with state $x = x(t)$, we denote the next state using $x' = x(t+1) = f(x)$. 
The simulations were conducted on a Windows 11 device equipped with an AMD Ryzen 9 4900HS Processor and 16GB of RAM. 
We used TSSOS \cite{wang2021tssos} in Julia to implement the SOS formulations provided in Section \ref{sec:synth}. 
A comparison table of our experiments alongside some key parameters can be found in Table~\ref{tab:compare}.
All computed certificates can be found in Appendix~\ref{appendix:casestudies}.

\subsection{Safety}
\label{subsec:casestudy_safety}
We illustrate the advantage of using VCCs for verifying safety on two 2D systems with states $x = [x_1, x_2]^T$: a simple model and the Van der Pol oscillator.
We used conditions \eqref{eq:vcc_safe1}-\eqref{eq:vcc_safe3} formulated as an SOS optimization problem, as described in Section \ref{subsec:find_safety}. 

\subsubsection{2D Simple System}
\label{subsec:casestudy_safety1}

The dynamics of our model are given by the following difference equations:
\begin{align*}
&\begin{cases}
    x_1' = x_2,\\
    x_2' = -x_1.
\end{cases}
\end{align*}
The state set, initial set, and unsafe set are given by $\Xx = [-4,4]^2$, $\Xx_0 = [0,0.5]\times[-3.5,-3]$, and $\Xx_{u} = ([-4,-1]\times[1,4]) \cup ([1,4]\times[-4,-1])$, respectively. 
We used $\eta_j$ as decision variables to minimize with a lower bound of $0.001$ for both unsafe regions in the search for the VCC.
We start from a linear function in two state variables $x,y$ (four scalar variables $x_1,x_2,y_1,y_2$) as the template for the VCC per Definition~\ref{def:vcc_safe}. 
We then increase these up to a maximum degree $degree_{max} = 7$ polynomial and a maximum number of $k_{max} = 2$ functions. 
The smallest degree of a standard closure certificate we could find was degree five. 
However, we were able to find a VCC using a degree three polynomial template with $k = 2$ and matrix $A = \big(
\begin{smallmatrix}
  0 & 1\\
  1 & 0
\end{smallmatrix}\big)$. 
Figure \ref{fig:casestudy_safety1} shows sampled runs of the dynamical system that confirm $\Xx_{u}$ is never visited.
Table~\ref{tab:compare} shows key comparisons of our results with other methods for the verification of safety. 
\begin{figure}[!t]
    \centering
    \epsfig{file=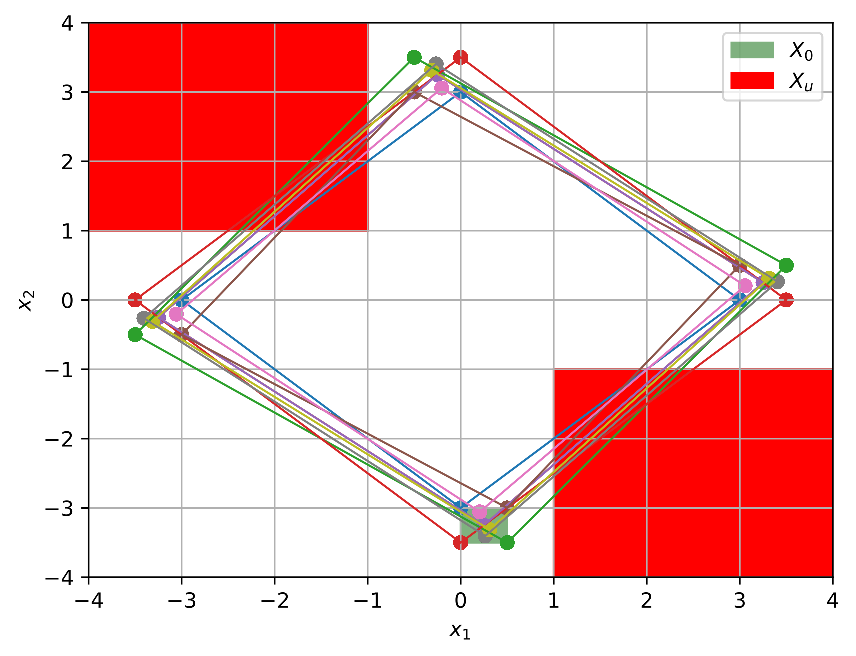, width=0.45\textwidth, keepaspectratio}
    \caption{Sampled state trajectories of simple dynamic model. The trajectory lines show jumps. The green and red shaded regions represent $\Xx_0$ and $\Xx_u$, respectively. Set $\Xx_u$ is never visited. Note that as we deal with discrete time systems, the states of the system jump and thus, the trajectories do not intersect with the unsafe region.}
    \label{fig:casestudy_safety1}
\end{figure}

\subsubsection{2D Van der Pol Oscillator}
\label{subsec:casestudy_safety2}

The model dynamics is given by the following difference equations:
\begin{align*}
&\begin{cases}
    x_1' = x_1 + Tx_2,\\
    x_2' = x_2 + T(-x_1 + u x_2(1-x_1^2)),
\end{cases}
\end{align*}
where $T = 0.1$ is the sampling time and $u = 0.4$.
The state set, initial set, and unsafe set are given by $\Xx = [-2.5,4]\times [-2.5,2.5]$, $\Xx_0 = [3,3.5]\times[1.5,2]$, and $\Xx_{u} = ([2,4]\times[-2.5,-1.5]) \cup ([-0.5,0.5]\times[-0.5,0.5]) \cup ([-2.5,-1.5]\times[1.5,2.5])$, respectively. 
We used $\eta_j$ as decision variables to minimize, with a lower bound of $0.001$ for all unsafe regions to search for the VCC.
We search for polynomials up to a maximum degree $degree_{max} = 6$ and a maximum number of $k_{max} = 3$ functions.
We were able to find a VCC of degree four of $k = 3$ functions with $A = \big(
\begin{smallmatrix}
  1 & 1 & 0\\
  1 & 0 & 0\\
  1 & 0 & 1
\end{smallmatrix}\big)
$ much faster than the nearest standard CC of degree five that we could find. 
See Table~\ref{tab:compare} for detailed comparisons.

Figure \ref{fig:casestudy_safety2} shows sampled runs of the dynamical system that confirm $\Xx_{u}$ is never visited.
\begin{figure}[!t]
    \centering
    \epsfig{file=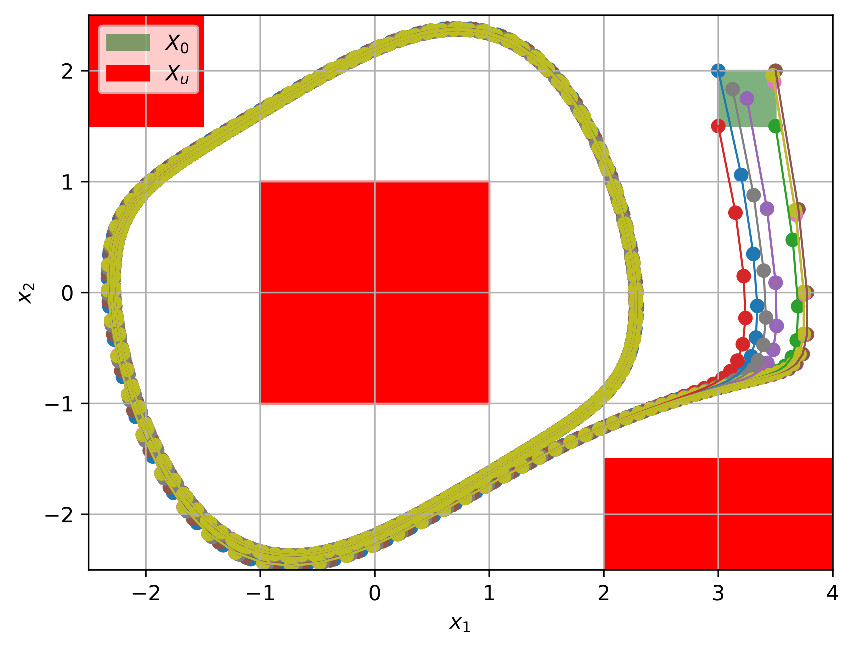, width=0.45\textwidth, keepaspectratio}
  \caption{Sampled state trajectories of Van der Pol oscillator model. The green and red shaded regions represent $\Xx_0$ and $\Xx_u$, respectively. Set $\Xx_u$ is never visited.}
  \label{fig:casestudy_safety2}
\end{figure}

\subsection{Persistence}
\label{subsec:casestudy_pers}

We experimentally demonstrate the utility of VCBRFs and VCCs to verify persistence in a 3D Kuramoto oscillator model with states $x = [x_1, x_2, x_3]^T$.

The dynamics of our model are given by the following difference equations:
\begin{align*}
&\begin{cases}
    x_1' = x_1 + TKsin(x_2-x_1) + cx_1^2 + (\Omega + 0.97),\\
    x_2' = x_2 + TK(sin(x_1-x_2) + sin(x_3-x_2)) + cx_2^2 \\
    \qquad + (\Omega + 0.97),\\
    x_3' = x_3 + TKsin(x_2-x_3) + cx_3^2 + (\Omega + 0.97),
\end{cases}
\end{align*}
where $T = 0.1s$ is the sampling time, $\Omega = 0.03$ is the natural frequency, $K = 0.06$ is the coupling strength, and $c = -0.532$. 
The state set, initial set, and finitely visited state set are given by $\Xx = [0,2]^3,\ \Xx_0 = [0,\frac{2\pi}{15}]^3$ and $\Xx_{VF} = ([0,0.7]^2\times[0,2]) \cup ([0,2]\times[1.45,2]^2)$, respectively. 
For both VCBRF and VCC search, we used the SOS optimization program as described in Section \ref{subsec:find_persistence}.
We set $\gamma_i = 1$ and $\rho_i = 1$ for VCC conditions. 
We used $sin(x) \approx x - \frac{x^3}{6}$ as the Taylor approximation for all the sine terms.
We defined $\eta_j$ as the decision variables to be minimized, with a lower bound of $0.001$.
For VCCs, we varied $k$ from $0$ to $k_{max} = 2$ and the degree of the polynomial template from $1$ to $degree_{max} = 4$, above which the SOS fails to compute due to the size of the SOS problem and device memory constraints. 
For VCBRFs, we set $k_{max} = 3$ and $degree_{max} = 8$.
We were unsuccessful in finding a standard closure certificate up to degree four. 
However, we found a VCBRF and VCC using a degree three polynomial template with $k = 2$ and $\eta_j = 0.001$. 
For VCBRF, we used $A_1 = \big(
\begin{smallmatrix}
  0 & 0\\
  1 & 0
\end{smallmatrix}\big)$, 
$A_2 = A_3 = \big(
\begin{smallmatrix}
  0 & 0\\
  0 & 0
\end{smallmatrix}\big)$.
For VCC, we used $A = \big(
\begin{smallmatrix}
  2 & 0\\
  0 & 1
\end{smallmatrix}\big)$. 
Observe that in the case of VCC here, the use of a diagonal matrix allows us to dedicate each function to each finitely visited region and obtain lower degree polynomial certificates.
See Table~\ref{tab:compare} for comparisons.

\subsection{LTL Specification}
\label{subsec:casestudy_ltl}

Here, we consider the following dynamics for a four state Lotka Volterra type prey-predator model:
\begin{align*}
&\begin{cases}
    x_1' = x_1 + T(r x_2 - (b_1+d_1) x_1 -h x_1 x_4)\\
    x_2' = x_2 + T(b_1 x_1 - b_2 x_2^2 - h x_2 x_4 - d_2 x_2)\\
    x_3' = x_3 + T(\alpha h x_1 x_4 + \beta h_2 x_2 x_4 - (n+d_3) x_3 -\xi x_2 x_4)\\
    x_4' = x_4 + T (n x_3 - d_4 x_4),
\end{cases}
\end{align*}
where $r = 1.6$, $b_1 = 0.3$, $b_2 = 0.3$, $h = 20$, $\alpha = 0.2$, $\beta = 0.2$, $\xi = 0.08$, $n = 0.3$, $d_1 = 0.08$, $d_2 = 0.06$, $d_3 = 0.5$, $d_4 = 0.5$ and $T = 0.01$.
The state set and initial set are given by $\Xx = [0, 7] \times [0, 6] \times [0, 8] \times [0, 4]$ and $\Xx_0 = [6.5,7] \times [5.5,6] \times [4.5,5] \times [3.5,4]$, respectively. 
We defined $\eta_j$ as the decision variables to be minimized, with a lower bound of $0.001$ for all accepting states.

Let the LTL specification to be verified be $\mathsf{F} \neg a \land \mathsf{F} \neg b \land \mathsf{FG} \neg c$.
This property requires that a system eventually visits the states labeled with atomic propositions $\neg a$ and $\neg b$ and finally stabilizes in the states labeled with the atomic proposition $\neg c$. The complement of this specification is denoted by an NBA $\Aa$ in Figure~\ref{fig:casestudy_ltl1}. 
We introduce an additional label: $\neg d$ for when the state is labeled with the complement of $\neg c$. Explicitly, the labeling map $\Ll : \Xx \rightarrow \Sigma$ is given by:
\[
\Ll(x) = 
\begin{cases}
    \neg a \quad \text{if } x \in [0,7] \times [0.5,2]\times [5.5,7.5]\times [0,4],\\
    \neg b \quad \text{if } x \in [0,7] \times [0,6]\times [1,3]\times [2,4],\\
    \neg c \quad \text{if } x \in [0,7] \times [0,6]\times [0,1]^2,\\
    \neg d \quad \text{otherwise}.
\end{cases}
\]
We utilize the appropriate LTL conditions expressed as an SOS optimization problem, as described in Section~\ref{subsec:find_ltl}.
We set $\gamma_i = 1$ and $\rho_i = 1$ for VCC conditions. 
Note that for VCCs, since we have three accepting states, when $k = 3$, we enforce condition~\eqref{eq:sos_vcc_ltl3} by dedicating each function to each accepting state. 
Otherwise, one of the functions will have the condition for more than one accepting state.
For VCBRFs, we searched for polynomials up to a maximum degree $degree_{max} = 7$ and a maximum number of functions $k_{max} = 3$.
We found a degree five VCBRF of $k = 2$ functions with $A_1 = A_2 = A_3 = \big(
\begin{smallmatrix}
  1 & 0\\
  1 & 1
\end{smallmatrix}\big)$.
We modified $degree_{max} = 4$ for VCC and found a degree three VCC of $k = 3$ functions with $A = \big(
\begin{smallmatrix}
  1 & 0 & 0\\
  1 & 1 & 0\\
  0 & 1 & 1
\end{smallmatrix}\big)
$. 
We could not search for CCs and VCCs above the degree $4$ due to the large size of the problem.
Table~\ref{tab:compare} shows comparisons of the results we obtained.

\begin{figure}[!t]
    \centering
    \begin{tikzpicture}[node distance =2cm]
    \node[initial, state, draw, initial text =,fill=blue!10!white] (1) at (0,0) {$q_1$};
    \node[accepting, state, fill = blue!10!white] (2) at (2,1.5) {$q_2$};
    \node[accepting, state, fill = blue!10!white] (3) at (2, 0) {$q_3$};
    \node[accepting, state, fill = blue!10!white] (4) at (2,-1.5) {$q_4$};
    \path[->]
    (1) edge[loop above] node[above]{$d$} (1)
    (1) edge[] node[above]{$a$} (2)
    (1) edge[] node[above]{$b$} (3)
    (1) edge[] node[above]{$c$} (4)
    (2) edge[loop right] node[right]{$a$} (2)
    (3) edge[loop right] node[right]{$b$} (3)
    (4) edge[bend left] node[left=0.1cm]{$d$} (1)
    (4) edge[loop right] node[right]{$c$} (4);   
    \end{tikzpicture}
    \caption{A (nondeterministic) B\"uchi automaton $\Aa$ for the LTL specification from Section~\ref{subsec:casestudy_ltl}. 
    The automata represents the LTL formula $\neg( \mathsf{F} \neg a \land \mathsf{F} \neg b \land \mathsf{FG} \neg c)$.}
    \label{fig:casestudy_ltl1}
\end{figure}
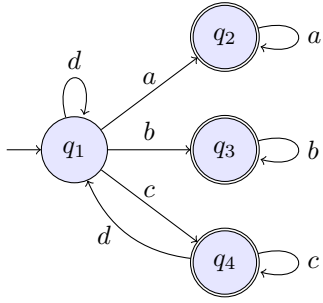

\begin{table}[!t]
\centering
\caption{Comparison of different certificates for the case studies.}
\begin{tabular}{ |p{0.21\linewidth}|p{0.27\linewidth}|p{0.17\linewidth}|p{0.2\linewidth}|  }
 \hline
System & Method & Polynomial Degree & Computation Times (s)\\
\hline
 \multirow{3}{*}{2D Simple}  & BC & $5$ &  $25.46$\\
 & CC & $5$ & $47.7$\\
 & VCC ($k = 2$) & $3$ & $18.75$\\
 \hline
 \multirow{3}{*}{Van der Pol}  & BC & $5$ &  $68.84$\\
 & CC & $5$ & $1263.30$\\
 & VCC ($k = 3$) & $4$ & $118.72$\\
 \hline
 \multirow{3}{*}{Kuramoto} & VCBRF ($k = 2$) & $3$ & $50.09$\\
 & CC & NF ($\leq 4$) & $-$\\
 & VCC ($k = 2$) & $3$ & $391.31$\\
 \hline
 \multirow{3}{*}{Prey-predator} & VCBRF ($k = 2$) & $5$ & $210.15$\\
 & CC & NF ($\leq 4$)& $-$\\
 & VCC ($k = 3$) & $3$ & $3002.51$\\
 \hline
\end{tabular}
{\raggedright \vspace{0.1em}* NF = Not Found \par}
\label{tab:compare}
\end{table}

\section{Conclusion}
We proposed two notions of vector certificates (VCBRFs and VCCs) that relax the standard conditions by finding multiple functions that together verify $\omega$-regular specifications of a given dynamical system.
We computed these certificates using SOS programming and demonstrated their potential computational advantages. 
Additionally, we present the benefit of VCCs over VCBRFs when the sets of interest are unions of multiple sets.
In future work, we plan to investigate how to effectively address bilinear conditions without having to experiment with the nonnegative matrix $A$ and potentially extend the formulation using a state-dependent nonnegative matrix $A(x,y)$.
Moreover, we intend to explore how one may use data-driven approaches to allow for a larger class of expressive functions to act as certificates.

\bibliographystyle{alpha}
\bibliography{ref.bib}

@inproceedings{prajna2004safety,
  title={Safety verification of hybrid systems using barrier certificates},
  author={Prajna, Stephen and Jadbabaie, Ali},
  booktitle={International Workshop on Hybrid Systems: Computation and Control},
  pages={477--492},
  year={2004},
publisher ={Springer}
}

@article{prajna2007convex,
author = {Prajna, Stephen and Rantzer, Anders},
year = {2007},
title = {Convex programs for temporal verification of nonlinear dynamical systems},
pages = {999-1021},
journal = {SIAM Journal on Control and Optimization},
}

@article{parrilo2003semidefinite,
  title={Semidefinite programming relaxations for semialgebraic problems},
  author={Parrilo, Pablo A},
  journal={Mathematical programming},
  volume={96},
  pages={293--320},
  year={2003},
  publisher={Springer}
}

@inproceedings{prajna2002sostools,
  title={Introducing SOSTOOLS: A general purpose sum of squares programming solver},
  author={Prajna, Stephen and Papachristodoulou, Antonis and Parrilo, Pablo A},
  booktitle={Proceedings of the 41st IEEE Conference on Decision and Control, 2002.},
  pages={741--746},
  year={2002},
publisher = {IEEE}
}

@article{vardi1994reasoning,
  title={Reasoning about infinite computations},
  author={Vardi, Moshe Y and Wolper, Pierre},
  journal={Information and computation},
  volume={115},
  number={1},
  pages={1--37},
  year={1994},
  publisher={Elsevier},
}

@article{moura2011smt,
  title={Satisfiability modulo theories: introduction and applications},
  author={De Moura, Leonardo and Bj{\o}rner, Nikolaj},
  journal={Communications of the ACM},
  volume={54},
  number={9},
  pages={69--77},
  year={2011},
  publisher={ACM New York, NY, USA}
}

@inproceedings{pnueli1977temporal,
  title={The temporal logic of programs},
  author={Pnueli, Amir},
  booktitle={18th Annual Symposium on Foundations of Computer Science},
  pages={46--57},
  year={1977},
  organization={IEEE}
}

@inproceedings{dimitrova2014deductive,
  title={Deductive control synthesis for alternating-time logics},
  author={Dimitrova, Rayna and Majumdar, Rupak},
  booktitle={2014 International Conference on Embedded Software (EMSOFT)},
  pages={1--10},
  year={2014},
  organization={IEEE}
}

@inproceedings{podelski2006model,
  title={Model checking of hybrid systems: From reachability towards stability},
  author={Podelski, Andreas and Wagner, Silke},
  booktitle={International Workshop on Hybrid Systems: Computation and Control},
  pages={507--521},
  year={2006},
  organization={Springer}
}

@inproceedings{podelski2004transition,
  title={Transition invariants},
  author={Podelski, Andreas and Rybalchenko, Andrey},
  booktitle={Proceedings of the 19th Annual IEEE Symposium on Logic in Computer Science, 2004.},
  pages={32--41},
  year={2004},
  organization={IEEE}
}

@article{vardi2005automata,
  title={An automata-theoretic approach to linear temporal logic},
  author={Vardi, Moshe Y},
  journal={Logics for concurrency: structure versus automata},
  pages={238--266},
  year={2005},
  publisher={Springer}
}

@inproceedings{nadali2024closure,
    author = {Nadali, Alireza and Murali, Vishnu and Trivedi, Ashutosh and Zamani, Majid} ,
    title = {Neural Closure Certificates},
    booktitle = {Proceedings of the AAAI Conference on Artificial Intelligence},
    year = {2024}
}

@inproceedings{murali2024closure,
  title={Closure certificates},
  author={Murali, Vishnu and Trivedi, Ashutosh and Zamani, Majid},
  booktitle={Proceedings of the 27th ACM International Conference on Hybrid Systems: Computation and Control},
  pages={1--11},
  year={2024}
}

@ARTICLE{oumer2024ibc,
  author={Oumer, Mohammed Adib and Murali, Vishnu and Trivedi, Ashutosh and Zamani, Majid},
  journal={IEEE Control Systems Letters}, 
  title={Safety Verification of Discrete-Time Systems via Interpolation-Inspired Barrier Certificates}, 
  year={2024},
  volume={8},
  number={},
  pages={3183-3188},
  doi={10.1109/LCSYS.2024.3521356}
}

@article{wang2021tssos,
  title={TSSOS: A moment-SOS hierarchy that exploits term sparsity},
  author={Wang, Jie and Magron, Victor and Lasserre, Jean-Bernard},
  journal={SIAM Journal on optimization},
  volume={31},
  number={1},
  pages={30--58},
  year={2021},
  publisher={SIAM}
}

@article{alpern1987recognizing,
  title={Recognizing safety and liveness},
  author={Alpern, Bowen and Schneider, Fred B},
  journal={Distributed computing},
  volume={2},
  number={3},
  pages={117--126},
  year={1987},
  publisher={Springer}
}

@article{yakubovich1971s,
  title={S-procedure in nolinear control theory},
  author={Yakubovich, Vladimir A},
  journal={Vestnik Leninggradskogo Universiteta, Ser. Matematika},
  pages={62--77},
  year={1971}
}

@inproceedings{sogokon2018vector,
  title={Vector barrier certificates and comparison systems},
  author={Sogokon, Andrew and Ghorbal, Khalil and Tan, Yong Kiam and Platzer, Andr{\'e}},
  booktitle={International Symposium on Formal Methods},
  pages={418--437},
  year={2018},
  organization={Springer}
}

@inproceedings{kong2013exponential,
  title={Exponential-condition-based barrier certificate generation for safety verification of hybrid systems},
  author={Kong, Hui and He, Fei and Song, Xiaoyu and Hung, William NN and Gu, Ming},
  booktitle={International Conference on Computer Aided Verification},
  pages={242--257},
  year={2013},
  organization={Springer}
}

@article{dai2017barrier,
  title={Barrier certificates revisited},
  author={Dai, Liyun and Gan, Ting and Xia, Bican and Zhan, Naijun},
  journal={Journal of Symbolic Computation},
  volume={80},
  pages={62--86},
  year={2017},
  publisher={Elsevier}
}

@inproceedings{berger2024cone,
  title={Cone-based abstract interpretation for nonlinear positive invariant synthesis},
  author={Berger, Guillaume and Ghanbarpour, Masoumeh and Sankaranarayanan, Sriram},
  booktitle={Proceedings of the 27th ACM International Conference on Hybrid Systems: Computation and Control},
  pages={1--16},
  year={2024}
}

@inproceedings{wang2021synthesizing,
  title={Synthesizing invariant barrier certificates via difference-of-convex programming},
  author={Wang, Qiuye and Chen, Mingshuai and Xue, Bai and Zhan, Naijun and Katoen, Joost-Pieter},
  booktitle={International Conference on Computer Aided Verification},
  pages={443--466},
  year={2021},
  organization={Springer}
}

@inproceedings{chatterjee2024sound,
  title={Sound and complete witnesses for template-based verification of LTL properties on polynomial programs},
  author={Chatterjee, Krishnendu and Goharshady, Amir and Goharshady, Ehsan and Karrabi, Mehrdad and {\v{Z}}ikeli{\'c}, {\DJ}or{\dj}e},
  booktitle={International Symposium on Formal Methods},
  pages={600--619},
  year={2024},
  organization={Springer}
}

@incollection{ramsey1987problem,
  title={On a problem of formal logic},
  author={Ramsey, Frank P},
  booktitle={Classic Papers in Combinatorics},
  pages={1--24},
  year={1987},
  publisher={Springer}
}

@misc{gurobi,
  author = {{Gurobi Optimization, LLC}},
  title = {{Gurobi Optimizer Reference Manual}},
  year = 2024,
  url = "https://www.gurobi.com"
}

\newpage
\appendix

\section{Case Study Results}
\label{appendix:casestudies}

VCC functions for verifying safety of 2D simple system.

\begin{align*}
    \Tt_1&(x,y) = -19.489x_1^{2}y_1 + 10.28x_1^{2}y_2 - 30.435x_1x_2y_1 - 32.771x_1x_2y_2 + 19.575x_2^{2}y_1 - 10.403x_2^{2}y_2 \\
    &- 0.015y_1^{3} + 0.149y_1^{2}y_2 - 0.155y_1y_2^{2} - 0.014y_2^{3} + 28.376x_1^{2} + 0.008x_1x_2 + 28.693x_2^{2} + 58.051y_1^{2}\\
    &+ 109.296y_1y_2 + 30.06y_2^{2} + 0.874y_1 + 0.039y_2 + 10.356,\\
    \Tt_2&(x,y) = 19.575x_1^{2}y_1 - 10.403x_1^{2}y_2 + 30.435x_1x_2y_1 + 32.771x_1x_2y_2 - 19.489x_2^{2}y_1 + 10.28x_2^{2}y_2 - 0.015y_1^{3}\\
    &+ 0.149y_1^{2}y_2 - 0.155y_1y_2^{2} - 0.014y_2^{3} + 28.693x_1^{2} - 0.008x_1x_2 + 28.376x_2^{2} + 58.051y_1^{2} + 109.296y_1y_2\\
    &+ 30.06y_2^{2} + 0.874y_1 + 0.039y_2 + 10.356.
\end{align*}

\begin{table}[H]
\centering
\caption{VCC function parameters for verifying safety of Van der Pol oscillator model.}
\begin{tabular}{|p{1em}|p{0.95\textwidth}|} 
  \hline
  $\mathbf{p}$ &
  [$1$, $x_{1}$, $x_{2}$, $y_{1}$, $y_{2}$, $x_{1}^{2}$, $x_{1}x_{2}$, $x_{1}y_{1}$, $x_{1}y_{2}$, $x_{2}^{2}$, $x_{2}y_{1}$, $x_{2}y_{2}$, $y_{1}^{2}$, $y_{1}y_{2}$, $y_{2}^{2}$, $x_{1}^{3}$, $x_{1}^{2}x_{2}$, $x_{1}^{2}y_{1}$, $x_{1}^{2}y_{2}$, $x_{1}x_{2}^{2}$, $x_{1}x_{2}y_{1}$, $x_{1}x_{2}y_{2}$, $x_{1}y_{1}^{2}$, $x_{1}y_{1}y_{2}$, $x_{1}y_{2}^{2}$, $x_{2}^{3}$, $x_{2}^{2}y_{1}$, $x_{2}^{2}y_{2}$, $x_{2}y_{1}^{2}$, $x_{2}y_{1}y_{2}$, $x_{2}y_{2}^{2}$, $y_{1}^{3}$, $y_{1}^{2}y_{2}$, $y_{1}y_{2}^{2}$, $y_{2}^{3}$, $x_{1}^{4}$, $x_{1}^{3}x_{2}$, $x_{1}^{3}y_{1}$, $x_{1}^{3}y_{2}$, $x_{1}^{2}x_{2}^{2}$, $x_{1}^{2}x_{2}y_{1}$, $x_{1}^{2}x_{2}y_{2}$, $x_{1}^{2}y_{1}^{2}$, $x_{1}^{2}y_{1}y_{2}$, $x_{1}^{2}y_{2}^{2}$, $x_{1}x_{2}^{3}$, $x_{1}x_{2}^{2}y_{1}$, $x_{1}x_{2}^{2}y_{2}$, $x_{1}x_{2}y_{1}^{2}$, $x_{1}x_{2}y_{1}y_{2}$, $x_{1}x_{2}y_{2}^{2}$, $x_{1}y_{1}^{3}$, $x_{1}y_{1}^{2}y_{2}$, $x_{1}y_{1}y_{2}^{2}$, $x_{1}y_{2}^{3}$, $x_{2}^{4}$, $x_{2}^{3}y_{1}$, $x_{2}^{3}y_{2}$, $x_{2}^{2}y_{1}^{2}$, $x_{2}^{2}y_{1}y_{2}$, $x_{2}^{2}y_{2}^{2}$, $x_{2}y_{1}^{3}$, $x_{2}y_{1}^{2}y_{2}$, $x_{2}y_{1}y_{2}^{2}$, $x_{2}y_{2}^{3}$, $y_{1}^{4}$, $y_{1}^{3}y_{2}$, $y_{1}^{2}y_{2}^{2}$, $y_{1}y_{2}^{3}$, $y_{2}^{4}$]$^T$
  \\
  \hline
  $\mathbf{c}_1$ & [$0.0$, $0.0$, $0.0$, $0.0$, $0.0$, $0.0$, $0.0$, $0.0$, $0.0$, $0.0$, $0.0$, $0.0$, $0.0$, $0.0$, $0.0$, $0.0$, $0.0$, $0.0$, $0.0$, $0.0$, $0.0$, $0.0$, $0.0$, $0.0$, $0.0$, $0.0$, $0.0$, $0.0$, $0.0$, $0.0$, $0.0$, $0.0$, $0.0$, $0.0$, $0.0$, $-0.002$, $0.001$, $0.005$, $-0.002$, $-0.003$, $-0.004$, $0.004$, $-0.004$, $0.005$, $-0.002$, $0.0$, $0.004$, $-0.001$, $0.003$, $-0.008$, $0.003$, $0.002$, $-0.005$, $0.005$, $-0.002$, $-0.001$, $0.0$, $0.0$, $-0.001$, $0.0$, $0.0$, $0.0$, $0.004$, $-0.002$, $0.002$, $-0.001$, $0.002$, $-0.004$, $0.002$, $-0.001$]$^T$\\
\hline
\end{tabular}
\label{tab:vcc_safety2}
\end{table}

\begin{table}[H]
\centering
\caption{VCBRF function parameters for verifying persistence of Kuramoto oscillator model.}
\begin{tabular}{|p{1em}|p{0.95\textwidth}|} 
  \hline
  $\mathbf{p}$ &
  [$1$, $x_{3}$, $x_{2}$, $x_{1}$, $x_{3}^{2}$, $x_{2}x_{3}$, $x_{2}^{2}$, $x_{1}x_{3}$, $x_{1}x_{2}$, $x_{1}^{2}$, $x_{3}^{3}$, $x_{2}x_{3}^{2}$, $x_{2}^{2}x_{3}$, $x_{2}^{3}$, $x_{1}x_{3}^{2}$, $x_{1}x_{2}x_{3}$, $x_{1}x_{2}^{2}$, $x_{1}^{2}x_{3}$, $x_{1}^{2}x_{2}$, $x_{1}^{3}$]$^T$
  \\
  \hline
  $\mathbf{c}_1$ & [$87.197$, $-5.008$, $-7.484$, $-5.315$, $2.49$, $-15.192$, $0.765$, $-16.002$, $-14.086$, $2.544$, $8.238$, $-0.235$, $2.224$, $7.181$, $-2.232$, $-2.564$, $-0.506$, $2.347$, $2.04$, $5.863$]$^T$\\ 
  \hline
  $\mathbf{c}_2$ & [$222.74$, $-3.532$, $-6.296$, $-4.677$, $11.546$, $-24.109$, $8.973$, $-26.59$, $-23.094$, $11.556$, $10.902$, $-1.474$, $2.724$, $8.874$, $-3.766$, $-6.252$, $-0.922$, $3.317$, $2.19$, $7.265$]$^T$\\ 
  \hline
\end{tabular}
\label{tab:vcbrf_persistence}
\end{table}

\begin{table}[H]
\centering
\caption{VCC function parameters for verifying persistence of Kuramoto oscillator model.}
\begin{tabular}{|p{1em}|p{0.95\textwidth}|} 
  \hline
  $\mathbf{p}$ &
  [$1$, $x_{1}$, $x_{2}$, $x_{3}$, $y_{1}$, $y_{2}$, $y_{3}$, $x_{1}^{2}$, $x_{1}x_{2}$, $x_{1}x_{3}$, $x_{1}y_{1}$, $x_{1}y_{2}$, $x_{1}y_{3}$, $x_{2}^{2}$, $x_{2}x_{3}$, $x_{2}y_{1}$, $x_{2}y_{2}$, $x_{2}y_{3}$, $x_{3}^{2}$, $x_{3}y_{1}$, $x_{3}y_{2}$, $x_{3}y_{3}$, $y_{1}^{2}$, $y_{1}y_{2}$, $y_{1}y_{3}$, $y_{2}^{2}$, $y_{2}y_{3}$, $y_{3}^{2}$, $x_{1}^{3}$, $x_{1}^{2}x_{2}$, $x_{1}^{2}x_{3}$, $x_{1}^{2}y_{1}$, $x_{1}^{2}y_{2}$, $x_{1}^{2}y_{3}$, $x_{1}x_{2}^{2}$, $x_{1}x_{2}x_{3}$, $x_{1}x_{2}y_{1}$, $x_{1}x_{2}y_{2}$, $x_{1}x_{2}y_{3}$, $x_{1}x_{3}^{2}$, $x_{1}x_{3}y_{1}$, $x_{1}x_{3}y_{2}$, $x_{1}x_{3}y_{3}$, $x_{1}y_{1}^{2}$, $x_{1}y_{1}y_{2}$, $x_{1}y_{1}y_{3}$, $x_{1}y_{2}^{2}$, $x_{1}y_{2}y_{3}$, $x_{1}y_{3}^{2}$, $x_{2}^{3}$, $x_{2}^{2}x_{3}$, $x_{2}^{2}y_{1}$, $x_{2}^{2}y_{2}$, $x_{2}^{2}y_{3}$, $x_{2}x_{3}^{2}$, $x_{2}x_{3}y_{1}$, $x_{2}x_{3}y_{2}$, $x_{2}x_{3}y_{3}$, $x_{2}y_{1}^{2}$, $x_{2}y_{1}y_{2}$, $x_{2}y_{1}y_{3}$, $x_{2}y_{2}^{2}$, $x_{2}y_{2}y_{3}$, $x_{2}y_{3}^{2}$, $x_{3}^{3}$, $x_{3}^{2}y_{1}$, $x_{3}^{2}y_{2}$, $x_{3}^{2}y_{3}$, $x_{3}y_{1}^{2}$, $x_{3}y_{1}y_{2}$, $x_{3}y_{1}y_{3}$, $x_{3}y_{2}^{2}$, $x_{3}y_{2}y_{3}$, $x_{3}y_{3}^{2}$, $y_{1}^{3}$, $y_{1}^{2}y_{2}$, $y_{1}^{2}y_{3}$, $y_{1}y_{2}^{2}$, $y_{1}y_{2}y_{3}$, $y_{1}y_{3}^{2}$, $y_{2}^{3}$, $y_{2}^{2}y_{3}$, $y_{2}y_{3}^{2}$, $y_{3}^{3}$]$^T$
  \\
\hline
\end{tabular}
\label{tab:vcc_persistence}
\end{table}

\newpage
\centering
\topcaption{VCC function parameters for verifying an LTL specification of a prey-predator model.}
\begin{supertabular}{|p{2em}|p{0.9\textwidth}|} 
  \hline
  $\mathbf{p}$ &
  [$1$, $x_{4}$, $x_{3}$, $x_{2}$, $x_{1}$, $x_{4}^{2}$, $x_{3}x_{4}$, $x_{3}^{2}$, $x_{2}x_{4}$, $x_{2}x_{3}$, $x_{2}^{2}$, $x_{1}x_{4}$, $x_{1}x_{3}$, $x_{1}x_{2}$, $x_{1}^{2}$, $x_{4}^{3}$, $x_{3}x_{4}^{2}$, $x_{3}^{2}x_{4}$, $x_{3}^{3}$, $x_{2}x_{4}^{2}$, $x_{2}x_{3}x_{4}$, $x_{2}x_{3}^{2}$, $x_{2}^{2}x_{4}$, $x_{2}^{2}x_{3}$, $x_{2}^{3}$, $x_{1}x_{4}^{2}$, $x_{1}x_{3}x_{4}$, $x_{1}x_{3}^{2}$, $x_{1}x_{2}x_{4}$, $x_{1}x_{2}x_{3}$, $x_{1}x_{2}^{2}$, $x_{1}^{2}x_{4}$, $x_{1}^{2}x_{3}$, $x_{1}^{2}x_{2}$, $x_{1}^{3}$, $x_{4}^{4}$, $x_{3}x_{4}^{3}$, $x_{3}^{2}x_{4}^{2}$, $x_{3}^{3}x_{4}$, $x_{3}^{4}$, $x_{2}x_{4}^{3}$, $x_{2}x_{3}x_{4}^{2}$, $x_{2}x_{3}^{2}x_{4}$, $x_{2}x_{3}^{3}$, $x_{2}^{2}x_{4}^{2}$, $x_{2}^{2}x_{3}x_{4}$, $x_{2}^{2}x_{3}^{2}$, $x_{2}^{3}x_{4}$, $x_{2}^{3}x_{3}$, $x_{2}^{4}$, $x_{1}x_{4}^{3}$, $x_{1}x_{3}x_{4}^{2}$, $x_{1}x_{3}^{2}x_{4}$, $x_{1}x_{3}^{3}$, $x_{1}x_{2}x_{4}^{2}$, $x_{1}x_{2}x_{3}x_{4}$, $x_{1}x_{2}x_{3}^{2}$, $x_{1}x_{2}^{2}x_{4}$, $x_{1}x_{2}^{2}x_{3}$, $x_{1}x_{2}^{3}$, $x_{1}^{2}x_{4}^{2}$, $x_{1}^{2}x_{3}x_{4}$, $x_{1}^{2}x_{3}^{2}$, $x_{1}^{2}x_{2}x_{4}$, $x_{1}^{2}x_{2}x_{3}$, $x_{1}^{2}x_{2}^{2}$, $x_{1}^{3}x_{4}$, $x_{1}^{3}x_{3}$, $x_{1}^{3}x_{2}$, $x_{1}^{4}$, $x_{4}^{5}$, $x_{3}x_{4}^{4}$, $x_{3}^{2}x_{4}^{3}$, $x_{3}^{3}x_{4}^{2}$, $x_{3}^{4}x_{4}$, $x_{3}^{5}$, $x_{2}x_{4}^{4}$, $x_{2}x_{3}x_{4}^{3}$, $x_{2}x_{3}^{2}x_{4}^{2}$, $x_{2}x_{3}^{3}x_{4}$, $x_{2}x_{3}^{4}$, $x_{2}^{2}x_{4}^{3}$, $x_{2}^{2}x_{3}x_{4}^{2}$, $x_{2}^{2}x_{3}^{2}x_{4}$, $x_{2}^{2}x_{3}^{3}$, $x_{2}^{3}x_{4}^{2}$, $x_{2}^{3}x_{3}x_{4}$, $x_{2}^{3}x_{3}^{2}$, $x_{2}^{4}x_{4}$, $x_{2}^{4}x_{3}$, $x_{2}^{5}$, $x_{1}x_{4}^{4}$, $x_{1}x_{3}x_{4}^{3}$, $x_{1}x_{3}^{2}x_{4}^{2}$, $x_{1}x_{3}^{3}x_{4}$, $x_{1}x_{3}^{4}$, $x_{1}x_{2}x_{4}^{3}$, $x_{1}x_{2}x_{3}x_{4}^{2}$, $x_{1}x_{2}x_{3}^{2}x_{4}$, $x_{1}x_{2}x_{3}^{3}$, $x_{1}x_{2}^{2}x_{4}^{2}$, $x_{1}x_{2}^{2}x_{3}x_{4}$, $x_{1}x_{2}^{2}x_{3}^{2}$, $x_{1}x_{2}^{3}x_{4}$, $x_{1}x_{2}^{3}x_{3}$, $x_{1}x_{2}^{4}$, $x_{1}^{2}x_{4}^{3}$, $x_{1}^{2}x_{3}x_{4}^{2}$, $x_{1}^{2}x_{3}^{2}x_{4}$, $x_{1}^{2}x_{3}^{3}$, $x_{1}^{2}x_{2}x_{4}^{2}$, $x_{1}^{2}x_{2}x_{3}x_{4}$, $x_{1}^{2}x_{2}x_{3}^{2}$, $x_{1}^{2}x_{2}^{2}x_{4}$, $x_{1}^{2}x_{2}^{2}x_{3}$, $x_{1}^{2}x_{2}^{3}$, $x_{1}^{3}x_{4}^{2}$, $x_{1}^{3}x_{3}x_{4}$, $x_{1}^{3}x_{3}^{2}$, $x_{1}^{3}x_{2}x_{4}$, $x_{1}^{3}x_{2}x_{3}$, $x_{1}^{3}x_{2}^{2}$, $x_{1}^{4}x_{4}$, $x_{1}^{4}x_{3}$, $x_{1}^{4}x_{2}$, $x_{1}^{5}$]$^T$
  \\
\hline
\end{supertabular}

\newpage
\centering
\topcaption{VCC function parameters for verifying an LTL specification of a prey-predator model.}
\begin{supertabular}{|p{2em}|p{0.9\textwidth}|} 
  \hline
  $\mathbf{p}$ &
  [$1$, $x_{1}$, $x_{2}$, $x_{3}$, $x_{4}$, $y_{1}$, $y_{2}$, $y_{3}$, $y_{4}$, $x_{1}^{2}$, $x_{1}x_{2}$, $x_{1}x_{3}$, $x_{1}x_{4}$, $x_{1}y_{1}$, $x_{1}y_{2}$, $x_{1}y_{3}$, $x_{1}y_{4}$, $x_{2}^{2}$, $x_{2}x_{3}$, $x_{2}x_{4}$, $x_{2}y_{1}$, $x_{2}y_{2}$, $x_{2}y_{3}$, $x_{2}y_{4}$, $x_{3}^{2}$, $x_{3}x_{4}$, $x_{3}y_{1}$, $x_{3}y_{2}$, $x_{3}y_{3}$, $x_{3}y_{4}$, $x_{4}^{2}$, $x_{4}y_{1}$, $x_{4}y_{2}$, $x_{4}y_{3}$, $x_{4}y_{4}$, $y_{1}^{2}$, $y_{1}y_{2}$, $y_{1}y_{3}$, $y_{1}y_{4}$, $y_{2}^{2}$, $y_{2}y_{3}$, $y_{2}y_{4}$, $y_{3}^{2}$, $y_{3}y_{4}$, $y_{4}^{2}$, $x_{1}^{3}$, $x_{1}^{2}x_{2}$, $x_{1}^{2}x_{3}$, $x_{1}^{2}x_{4}$, $x_{1}^{2}y_{1}$, $x_{1}^{2}y_{2}$, $x_{1}^{2}y_{3}$, $x_{1}^{2}y_{4}$, $x_{1}x_{2}^{2}$, $x_{1}x_{2}x_{3}$, $x_{1}x_{2}x_{4}$, $x_{1}x_{2}y_{1}$, $x_{1}x_{2}y_{2}$, $x_{1}x_{2}y_{3}$, $x_{1}x_{2}y_{4}$, $x_{1}x_{3}^{2}$, $x_{1}x_{3}x_{4}$, $x_{1}x_{3}y_{1}$, $x_{1}x_{3}y_{2}$, $x_{1}x_{3}y_{3}$, $x_{1}x_{3}y_{4}$, $x_{1}x_{4}^{2}$, $x_{1}x_{4}y_{1}$, $x_{1}x_{4}y_{2}$, $x_{1}x_{4}y_{3}$, $x_{1}x_{4}y_{4}$, $x_{1}y_{1}^{2}$, $x_{1}y_{1}y_{2}$, $x_{1}y_{1}y_{3}$, $x_{1}y_{1}y_{4}$, $x_{1}y_{2}^{2}$, $x_{1}y_{2}y_{3}$, $x_{1}y_{2}y_{4}$, $x_{1}y_{3}^{2}$, $x_{1}y_{3}y_{4}$, $x_{1}y_{4}^{2}$, $x_{2}^{3}$, $x_{2}^{2}x_{3}$, $x_{2}^{2}x_{4}$, $x_{2}^{2}y_{1}$, $x_{2}^{2}y_{2}$, $x_{2}^{2}y_{3}$, $x_{2}^{2}y_{4}$, $x_{2}x_{3}^{2}$, $x_{2}x_{3}x_{4}$, $x_{2}x_{3}y_{1}$, $x_{2}x_{3}y_{2}$, $x_{2}x_{3}y_{3}$, $x_{2}x_{3}y_{4}$, $x_{2}x_{4}^{2}$, $x_{2}x_{4}y_{1}$, $x_{2}x_{4}y_{2}$, $x_{2}x_{4}y_{3}$, $x_{2}x_{4}y_{4}$, $x_{2}y_{1}^{2}$, $x_{2}y_{1}y_{2}$, $x_{2}y_{1}y_{3}$, $x_{2}y_{1}y_{4}$, $x_{2}y_{2}^{2}$, $x_{2}y_{2}y_{3}$, $x_{2}y_{2}y_{4}$, $x_{2}y_{3}^{2}$, $x_{2}y_{3}y_{4}$, $x_{2}y_{4}^{2}$, $x_{3}^{3}$, $x_{3}^{2}x_{4}$, $x_{3}^{2}y_{1}$, $x_{3}^{2}y_{2}$, $x_{3}^{2}y_{3}$, $x_{3}^{2}y_{4}$, $x_{3}x_{4}^{2}$, $x_{3}x_{4}y_{1}$, $x_{3}x_{4}y_{2}$, $x_{3}x_{4}y_{3}$, $x_{3}x_{4}y_{4}$, $x_{3}y_{1}^{2}$, $x_{3}y_{1}y_{2}$, $x_{3}y_{1}y_{3}$, $x_{3}y_{1}y_{4}$, $x_{3}y_{2}^{2}$, $x_{3}y_{2}y_{3}$, $x_{3}y_{2}y_{4}$, $x_{3}y_{3}^{2}$, $x_{3}y_{3}y_{4}$, $x_{3}y_{4}^{2}$, $x_{4}^{3}$, $x_{4}^{2}y_{1}$, $x_{4}^{2}y_{2}$, $x_{4}^{2}y_{3}$, $x_{4}^{2}y_{4}$, $x_{4}y_{1}^{2}$, $x_{4}y_{1}y_{2}$, $x_{4}y_{1}y_{3}$, $x_{4}y_{1}y_{4}$, $x_{4}y_{2}^{2}$, $x_{4}y_{2}y_{3}$, $x_{4}y_{2}y_{4}$, $x_{4}y_{3}^{2}$, $x_{4}y_{3}y_{4}$, $x_{4}y_{4}^{2}$, $y_{1}^{3}$, $y_{1}^{2}y_{2}$, $y_{1}^{2}y_{3}$, $y_{1}^{2}y_{4}$, $y_{1}y_{2}^{2}$, $y_{1}y_{2}y_{3}$, $y_{1}y_{2}y_{4}$, $y_{1}y_{3}^{2}$, $y_{1}y_{3}y_{4}$, $y_{1}y_{4}^{2}$, $y_{2}^{3}$, $y_{2}^{2}y_{3}$, $y_{2}^{2}y_{4}$, $y_{2}y_{3}^{2}$, $y_{2}y_{3}y_{4}$, $y_{2}y_{4}^{2}$, $y_{3}^{3}$, $y_{3}^{2}y_{4}$, $y_{3}y_{4}^{2}$, $y_{4}^{3}$]$^T$
  \\
\hline
\end{supertabular}

\end{document}